\documentclass[10pt,twocolumn]{article}

\usepackage[margin=0.75in]{geometry}
\usepackage{graphicx}
\usepackage{amsmath,amssymb,mathtools,bm}
\usepackage{amsthm}
\usepackage{enumitem}
\usepackage{cite}
\usepackage[hidelinks]{hyperref}
\usepackage{microtype}

\newcommand{\R}{\mathbb{R}}
\newcommand{\dif}{\mathrm{d}}
\newcommand{\e}{\mathrm{e}}
\newcommand{\argmin}{\operatorname*{arg\,min}}

\newcommand{\bigO}{\mathcal{O}}
\newcommand{\esc}{\mathrm{esc}}
\newcommand{\pred}{\mathrm{p}}
\newcommand{\hp}{\mathrm{HP}}

\newtheorem{theorem}{Theorem}

\newtheorem{proposition}[theorem]{Proposition}

\theoremstyle{definition}

\newtheorem{assumption}{Assumption}

\theoremstyle{remark}
\newtheorem{remark}{Remark}

\graphicspath{{figs/}}

\title{Safe Newton-Based Extremum Seeking for Static Maps with Delayed Output Measurements}

\author{%
Azad Ghaffari\\
\small School of Engineering and Computing, Christopher Newport University\\
\small Newport News, VA 23606, USA\\
\small \texttt{azad.ghaffari@cnu.edu}
\and
Tiago Roux Oliveira\\
\small Department of Electronics and Telecommunication Engineering\\
\small State University of Rio de Janeiro\\
\small Rio de Janeiro 20550-013, Brazil\\
\small \texttt{tiagoroux@uerj.br}
}

\date{}

\begin{document}
\maketitle

\begin{abstract}
This work presents a delayed safe Newton-based extremum seeking (SANES) framework for minimizing an unknown static map subject to an unknown safety constraint. The objective and safety measurements are assumed to be affected by the same constant time delay. To compensate for delayed measurements, a model-free predictor is developed to construct the quantities required for optimization, including the nominal Newton-based extremum-seeking control input and the gradient information used to formulate control Lyapunov function (CLF) and control barrier function (CBF) conditions. Robust CLF--CBF quadratic programs (QPs), subject to parameter-update constraints, are then formulated to account explicitly for derivative-estimation and prediction errors. For the delay-free case, robustness margins are derived from bounds on the extremum-seeking estimation errors, whereas for the delayed case, the margins incorporate both estimation and prediction errors. Practical stability of the nominal Newton-based extremum-seeking dynamics is established directly through a Lyapunov analysis, thereby providing convergence guarantees over the admissible parameter set without relying exclusively on local averaging arguments. Robust CLF and CBF conditions are subsequently derived to establish practical convergence and forward invariance of a robust subset of the prescribed safe set. A numerical case study demonstrates the effectiveness of the proposed SANES framework in achieving constrained optimization despite unknown objective and safety maps and delayed measurements.
\end{abstract}


\section{Introduction}

%\cite{abel2021safety, abel2023prescribed, abel2020constrained, zhang2022control, yang2026overview, xie2025cbf, williams2026local, williams2023semi, oliveira2016extremum, oliveira2015newton, jankovic2018control, jankovic2018robust, ghaffari2012multivariable, ghaffari2025second}

Extremum seeking (ES) is a model-free methodology for real-time optimization of an unknown performance map using output measurements and probing signals. Since the first rigorous stability analysis of ES \cite{KrsticWang2000}, the field has progressed from local to semiglobal designs \cite{TanNesicMareels2006}, Lie-bracket formulations \cite{DurrEtAl2013, LabarEtAl2019}, and delay-compensating architectures \cite{oliveira2015newton,SIAM_ES_PDE_book,oliveira2016extremum}. Newton-based ES is particularly attractive for multivariable optimization because gradient and Hessian estimates generate a curvature-normalized search direction, reducing sensitivity to the conditioning of the unknown map. This approach was established in~\cite{ghaffari2012multivariable}, with alternative second-order constructions in~\cite{ghaffari2025second}.

Optimization, however, must often be performed while satisfying safety constraints throughout the search. Constrained ES has been addressed through barrier mechanisms and constrained optimization \cite{DeHaanGuay2005}, as well as sampled-data formulations for analytically unknown but measurable constraints \cite{HazelegerEtAl2022}. Control barrier functions (CBFs), combined with control Lyapunov functions (CLFs) through quadratic programs (QPs), provide a systematic means to reconcile safety and performance \cite{jankovic2018robust,xie2025cbf,yang2026overview,zhang2022control,AmesEtAl2017}. When the safety map is itself unknown, however, the derivative information required by a conventional CBF is unavailable. Recent safe-ES designs overcome this difficulty using measured safety information: practical safety was considered in \cite{williams2023semi}, semiglobal safety filtering with unknown CBFs in \cite{WilliamsKrsticScheinker2025}, and assignable attraction to the safe set in \cite{williams2026local}.

A separate challenge arises from delayed measurements, since demodulation then extracts information associated with a past optimization variable. Newton-based ES under actuator and sensor delays was studied in \cite{oliveira2015newton}, predictor feedback for delayed static maps in \cite{oliveira2016extremum}, and a broader treatment of ES through delays and distributed-parameter dynamics in \cite{SIAM_ES_PDE_book}. Related time-delay analyses appear in \cite{ZhangFridman2023,YangFridman2023}. For safety, the issue is more delicate: a barrier inequality evaluated from delayed data certifies a past point, whereas safety must hold for the current parameter. Predictor-based CBF constructions for input-delay systems were developed in \cite{abel2020constrained,abel2021safety,jankovic2018control}, while robustness to uncertainty in barrier conditions was considered in \cite{abel2023prescribed,jankovic2018robust,zhang2022control}. These approaches, however, rely on model information unavailable in the model-free setting considered here.

Indeed, consider the simultaneous measurements
\begin{equation}
    y(t)=\ell\big(\theta(t-D)\big),\qquad
    z(t)=h\big(\theta(t-D)\big),
\end{equation}
where both the objective map $\ell$ and the safety map $h$ are unknown. The Newton direction requires the current gradient and Hessian of $\ell$, while the CBF requires the current value and gradient of $h$; only delayed scalar measurements of both maps are available. Thus, combining a delayed ES controller with an existing safety filter is insufficient because the delay affects precisely the quantities required for both optimization and safety certification.

To the best of the authors' knowledge, existing results do not address this combination of features: \textit{multivariable Newton-based extremum seeking, an unknown safety map, delayed objective and safety measurements, and rigorous safety guarantees based on explicit prediction-error margins}. Safe-ES results \cite{williams2023semi,williams2026local,WilliamsKrsticScheinker2025} do not compensate delayed measurements in the safety certificate; delayed ES designs \cite{oliveira2015newton,SIAM_ES_PDE_book,oliveira2016extremum,ZhangFridman2023,YangFridman2023} do not enforce an unknown CBF constraint; and delay-compensated CBF approaches \cite{abel2020constrained,abel2021safety,jankovic2018control} require unavailable system information. This gap motivates the present work.

This paper develops a \emph{safe Newton-based extremum seeking (SANES)}  framework for an unknown static objective map subject to an unknown safety constraint, with both measurements affected by the same known constant delay. A model-free predictor advances the delayed derivative estimates and safety information to the current parameter location, while estimation and prediction errors are propagated explicitly into robust CLF and CBF inequalities in a QP. 

The main contributions are as follows.

\begin{enumerate}
\item[\textbf{(a)}] \emph{A safe multivariable Newton-based ES architecture using only objective and safety measurements.}\\

\noindent Perturbation-based estimates replace the unavailable objective gradient and Hessian, as well as the barrier gradient. A CLF--CBF QP modifies the nominal Newton direction subject to convergence, safety, and parameter-rate constraints, while explicitly distinguishing convergence of the slow ES center from safety of the dithered parameter.\\

\item[\textbf{(b)}] \emph{A direct practical-stability analysis of Newton-based ES with estimation errors.}\\

\noindent Bounds on the gradient, Hessian, and inverse-Hessian estimation errors are incorporated directly into a Lyapunov analysis. The resulting ultimate bound depends explicitly on the probing amplitude and time-scale-separation error, establishing practical convergence for persistent probing.\\

\item[\textbf{(c)}] \emph{A model-free prediction mechanism for the quantities required simultaneously by optimization and safety.}\\

\noindent For $D>0$, Taylor-based predictors are constructed for the objective gradient and Hessian, Newton direction, barrier value, and barrier gradient using the known displacement of the ES center and derivative information generated from measurements. Their errors are bounded in terms of the delay, parameter-update rate, probing parameters, and estimation errors.\\

\item[\textbf{(d)}] \emph{Robust CLF--CBF inequalities that explicitly compensate prediction uncertainty.}\\

\noindent Estimation and prediction-error bounds yield strengthened CLF and CBF inequalities. The CLF condition ensures practical descent, while the hard CBF constraint protects safety despite errors in the predicted barrier value and gradient; only the convergence condition is relaxed when the objectives compete.\\

\item[\textbf{(e)}] \emph{Closed-loop guarantees of practical optimization and forward invariance under delayed measurements.}\\

\noindent Under the stated feasibility and error conditions, the parameter approaches an explicitly characterized neighborhood of the unknown optimizer while a robust subset of the safe set remains forward invariant. The ultimate bound separates prediction error from the unavoidable probing amplitude, and the SANES QP admits a unique solution whenever the hard safety and parameter-update constraints are compatible.
\end{enumerate}

The proposed methodology is illustrated through a source-seeking obstacle-avoidance problem with delayed measurements, where SANES approaches the unknown optimizer while maintaining the barrier condition. The example also illustrates the roles of the CLF and CBF robustness margins in the practical convergence neighborhood and obstacle clearance.

The remainder of the paper is organized as follows. Section~\ref{secao2} presents preliminaries and introduces the
delayed static-map optimization problem, the safety constraint, and an ideal
Newton-based CLF--CBF quadratic program. Section~\ref{secao3} develops the delay-free
SANES formulation and establishes the required estimation, practical-stability,
and forward-invariance properties. Section~\ref{secao4} constructs the model-free
predictors and develops the robust SANES controller for delayed measurements.
Section~\ref{secao5} presents the obstacle-avoidance example, and Section~\ref{secao6} concludes the
paper.

\section{Problem Statement and Assumptions}\label{secao2}

\subsection{Preliminary Definitions and Notation}

Let $\mathbb{R}_+:=[0,\infty)$.
For a vector $x\in\mathbb{R}^n$, $\|x\|$ denotes its Euclidean norm, while
$|\xi|$ denotes the absolute value of a scalar $\xi\in\mathbb{R}$. For a matrix
$A$, $\|A\|$ denotes the induced matrix $2$-norm. The symbols
$\preceq$ and $\prec$ denote the Loewner partial ordering on symmetric
matrices.

A continuous function
$\gamma:\mathbb{R}_+\rightarrow\mathbb{R}_+$ is said to belong to class-$\mathcal{K}$ if it is strictly increasing and satisfies $\gamma(0)=0$.
A continuous function $\alpha:\mathcal{I}\rightarrow\mathbb{R}$, where
$\mathcal{I}\subset\mathbb{R}$ is an interval containing the origin, is
called an extended class-$\mathcal{K}$ function if it is strictly increasing
and satisfies $\alpha(0)=0$.

Throughout the paper, gradients are represented as column vectors. In
particular,
\begin{equation}
    \nabla\ell(\theta):= \left[
        \frac{\partial\ell(\theta)}{\partial\theta_1}~\ldots~
        \frac{\partial\ell(\theta)}{\partial\theta_n}\right]^\top \in\mathbb{R}^n.
\end{equation}
The Hessian of $\ell$ is denoted by
\begin{equation}
    \nabla^2\ell(\theta):=\left[
        \frac{\partial^2\ell(\theta)}
        {\partial\theta_i\partial\theta_j}
    \right]_{i,j=1}^n
    \in\mathbb{R}^{n\times n}.
\end{equation}
More generally, for $k\geq3$, $\nabla^k\ell(\theta)$ denotes the
$k$th-order derivative of $\ell$ with respect to $\theta$, interpreted as a
$k$-linear map.

For a scalar-valued three-times continuously differentiable function
$\phi:\Phi\rightarrow\mathbb{R}$, the induced norm of its third-order
derivative is defined by
\begin{equation}
\label{eq:third_derivative_norm}
    \|\nabla^3\phi(\theta)\| := \sup_{\|u\|=\|v\|=\|w\|=1}
    \left|\sum_{i=1}^n\sum_{j=1}^n\sum_{k=1}^n T_{ijk}(\theta)u_iv_jw_k\right|,
\end{equation}
where $T_{ijk}(\theta)={\partial^3 \phi(\theta)}/{\partial\theta_i\partial\theta_j\partial\theta_k}$. 

\subsection{Delayed Static Map}

Consider the following static map subject to a constant delay:
\begin{equation}
    y(t)=\ell\bigl(\theta(t-D)\bigr),
    \label{eq:delayed_static_map}
\end{equation}
where $y(t)\in\mathcal{Y}\subset\mathbb{R}$ is the measured system output,
$\theta(t)\in\Theta\subset\mathbb{R}^n$ is the adjustable parameter, and
$D>0$ is a known constant delay. Let
$O_\Theta\subset\mathbb{R}^n$ be an open set containing $\Theta$, and let
$\ell:O_\Theta\rightarrow\mathbb{R}$ be an unknown static map such that $\ell(\Theta)\subseteq\mathcal{Y}$.
The sets $\Theta$ and $\mathcal{Y}$ denote the admissible parameter and
output sets, respectively. Although $\ell$ is memoryless, the delayed
input-output operator in \eqref{eq:delayed_static_map} is not memoryless.

\begin{assumption}
\label{assumption1}
The admissible parameter set $\Theta\subset\mathbb{R}^n$ is compact and
convex and has nonempty interior. The parameter-update law and the
admissible control set are assumed to ensure that
$\theta(t)\in\Theta$ for all times under consideration. 
The map $\ell:O_\Theta\rightarrow\mathbb{R}$ is three times continuously
differentiable. There exist constants $0<{m_2}\leq {M_2}<\infty$ and $0<M_3<\infty$ such
that
\begin{align}
\label{eq:hessian_bounds}    {m_2} I \preceq\nabla^2\ell(\theta)
    \preceq {M_2} I,\\
\label{eq:third_derivative_bound}    \left\|\nabla^3\ell(\theta)\right\| \le M_3,
\end{align}
for all $\theta\in\Theta$, where  $I\in\R^{n\times n}$ is the identity matrix. Consequently,
$\ell$ admits a unique minimizer over $\Theta$. This minimizer is assumed
to lie in the interior of $\Theta$, i.e.,
\begin{equation}
    \theta^\ast  := \operatorname*{arg\,min}_{\theta\in\Theta}\ell(\theta)
    \in\operatorname{int}(\Theta), \qquad \nabla\ell(\theta^\ast)=0.
    \label{eq:optimality_condition}
\end{equation}
\end{assumption}

\begin{remark}
The compactness of $\Theta$ and the continuity of $\ell$ guarantee the
existence of a minimizer of $\ell$ over $\Theta$. The lower Hessian bound
in \eqref{eq:hessian_bounds}, together with the convexity of $\Theta$,
guarantees its uniqueness. The additional requirement
$\theta^\ast\in\operatorname{int}(\Theta)$ ensures that the unconstrained
first-order optimality condition
$\nabla\ell(\theta^\ast)=0$ holds. Moreover, because $\ell\in C^3(O_\Theta)$ and $\Theta$ is compact, the
continuity of $\nabla^3\ell$ implies the existence of a finite constant
$M_3$ satisfying \eqref{eq:third_derivative_bound}. The bound is stated
explicitly to facilitate the subsequent Taylor-remainder analysis.
\end{remark}

Let $\ell^\ast:=\ell(\theta^\ast)$ denote the minimum value of the objective map, and define the shifted
objective function
\begin{equation}
    L(\theta):=\ell(\theta)-\ell^\ast.
    \label{eq:shifted_objective}
\end{equation}
It follows from Assumption~\ref{assumption1} that $L(\theta)\geq0$, $L(\theta^\ast)=0$, and $\nabla L(\theta)=\nabla\ell(\theta)$ for all $\theta\in\Theta$. Moreover, $L\in C^3(O_\Theta)$ and satisfies the same Hessian and
third-order derivative bounds as $\ell$. In particular, $L$ is
${m_2}$-strongly convex and its gradient is ${M_2}$-Lipschitz on
$\Theta$. Hence, the measured objective output can be written as
\begin{align}
\label{eq:shifted_output}    
y(t)  	&= \ell^\ast+L\bigl(\theta(t-D)\bigr).   
\end{align}

Define the Hessian of the objective map at the optimizer as
\begin{equation}
    H^\ast:=\nabla^2\ell(\theta^\ast).
    \label{eq:optimal_hessian}
\end{equation}
By Assumption~\ref{assumption1}, ${m_2} I \preceq H^\ast \preceq {M_2} I$.
Therefore, $H^\ast$ is symmetric positive definite and invertible, and
$\|(H^\ast)^{-1}\| \leq {m_2}^{-1}$.
Furthermore, the strong-convexity condition in
Assumption~\ref{assumption1} implies that
%\begin{equation}
%    \bigl(
%        \nabla\ell(\theta)-\nabla\ell(\theta^\ast)
%    \bigr)^\top
%    (\theta-\theta^\ast)
%    \geq
%    {m_2}\|\theta-\theta^\ast\|^2,
%    \qquad
%    \forall\,\theta\in\Theta.
%    \label{eq:strong_monotonicity_gradient}
%\end{equation}
%Since $\nabla\ell(\theta^\ast)=0$, the Cauchy--Schwarz inequality applied to
%\eqref{eq:strong_monotonicity_gradient} gives
%\begin{equation}
%    \|\nabla\ell(\theta)\|
%    \|\theta-\theta^\ast\|
%    \geq
%    {m_2}\|\theta-\theta^\ast\|^2.
%\end{equation}
%For $\theta\neq\theta^\ast$, division by
%$\|\theta-\theta^\ast\|>0$ yields
%\[
%    \|\nabla\ell(\theta)\|
%    \geq
%    {m_2}\|\theta-\theta^\ast\|.
%\]
%The same inequality holds trivially at $\theta=\theta^\ast$. Therefore,
%\begin{equation}
%    \|\nabla\ell(\theta)\|
%    \geq
%    {m_2}\|\theta-\theta^\ast\|,
%    \qquad
%    \forall\,\theta\in\Theta.
%    \label{eq:gradient_lower_bound}
%\end{equation}
%In addition, the upper Hessian bound in \eqref{eq:hessian_bounds} implies
%\[
%    \|\nabla\ell(\theta)\|
%    \leq
%    {M_2}\|\theta-\theta^\ast\|,
%    \qquad
%    \forall\,\theta\in\Theta.
%\]
\begin{equation}
    {m_2}\|\theta-\theta^\ast\| \leq \|\nabla\ell(\theta)\| \leq {M_2}\|\theta-\theta^\ast\|,
    \quad \forall\,\theta\in\Theta.
\end{equation}

In particular,
\begin{equation}
    \|\nabla\ell(\theta)\|\leq\epsilon \quad\Longrightarrow\quad
    \|\theta-\theta^\ast\| \leq \frac{\epsilon}{{m_2}},
    \label{eq:gradient_to_parameter_bound}
\end{equation}
for any $\varepsilon\geq0$. This relation will subsequently be used to
translate bounds on gradient-estimation and gradient-prediction errors into
practical convergence bounds on $\theta$.

\subsection{Control Objective and Safety Constraint}

The control objective is to drive the parameter $\theta(t)$ toward the
unknown optimizer $\theta^\ast$ using the delayed measurement
$y(t)=\ell(\theta(t-D))$, while ensuring that the parameter trajectory
remains in a prescribed safe subset of $\Theta$. Safety information is
provided through the delayed measurement
\begin{equation}
    z(t)=h\bigl(\theta(t-D)\bigr),
    \label{eq:safety_measurement}
\end{equation}
where $h:O_\Theta\rightarrow\mathbb{R}$ is a barrier function.

A maximization problem can be treated by minimizing the negative objective
map, provided that the sign-adjusted objective satisfies the regularity and
strong-convexity conditions in Assumption~\ref{assumption1}.

\begin{assumption}
\label{assumption2}
The map $h:O_\Theta\rightarrow\mathbb{R}$ is three times continuously
differentiable. The safe set is defined as
\begin{equation}
    \Theta_0 := \left\{ \theta\in\Theta: h(\theta)\geq0 \right\},
    \label{eq:safe_set}
\end{equation}
and the corresponding unsafe set is
\begin{equation}
    \Theta_{\mathrm{u}} := \Theta\setminus\Theta_0 = \left\{ \theta\in\Theta: h(\theta)<0
    \right\}.
\end{equation}
Let $\mathcal{I}_h\subset\mathbb{R}$ be an interval containing
$h(\Theta)$ and all predicted barrier values considered in the subsequent
analysis. Assume that there exists an extended class-$\mathcal{K}$ function
$\alpha:\mathcal{I}_h\rightarrow\mathbb{R}$ that is Lipschitz continuous on
$\mathcal{I}_h$. Thus, there exists a constant $C_\alpha>0$ such that
\begin{equation}
\label{eq:alphaLpsch}
    |\alpha(r_1)-\alpha(r_2)| \leq C_\alpha|r_1-r_2|, \qquad \forall\,r_1,r_2\in\mathcal{I}_h.
\end{equation}
Moreover, $h$ is a control barrier function (CBF) for the parameter dynamics
\begin{equation}
    \dot{\theta}(t)=v(t), \qquad v(t)\in\mathcal{V},
\end{equation}
in the sense that, for every $\theta$ in an open neighborhood of
$\Theta_0$ relative to $\Theta$, there exists an admissible control input
$v\in\mathcal{V}$ satisfying
\begin{equation}
    \left(\nabla h(\theta)\right)^\top v \geq -\alpha\bigl(h(\theta)\bigr).
    \label{eq:cbf_condition}
\end{equation}
%The corresponding CBF constraint is assumed to remain feasible along the
%closed-loop trajectory.
\end{assumption}

It follows from Assumption~\ref{assumption2} and the compactness of
$\Theta$ that there exist finite constants
$M_{h,1},M_{h,2},M_{h,3}>0$ such that
\begin{align}
\|\nabla h(\theta)\| &\leq M_{h,1}, \\
\|\nabla^2 h(\theta)\| &\leq M_{h,2}, \\
\|\nabla^3 h(\theta)\| &\leq M_{h,3},
\end{align}
for all $\theta\in\Theta$. %Indeed, these bounds follow from the continuity
%of the corresponding derivatives of $h$ on the compact set $\Theta$.
Furthermore, since $\Theta$ is convex, the boundedness of $\nabla h$
implies that $h$ is Lipschitz continuous on $\Theta$. In particular,
\begin{equation}
|h(\theta_1)-h(\theta_2)| \leq M_{h,1}\|\theta_1-\theta_2\|, \quad \forall~\theta_1,\theta_2\in\Theta.
\label{eq:hLipschitz}
\end{equation}
Similarly, the boundedness of $\nabla^2 h$ implies that $\nabla h$ is
Lipschitz continuous on $\Theta$, with
\begin{equation}
\|\nabla h(\theta_1)-\nabla h(\theta_2)\|
\leq M_{h,2}\|\theta_1-\theta_2\|,~
\forall~\theta_1,\theta_2\in\Theta.
\end{equation}

The set $\mathcal{V}\subset\mathbb{R}^n$ denotes the admissible control set.
Unless stated otherwise, $\mathcal{V}$ is assumed to be nonempty, compact,
and convex. Consequently, there exists a finite constant vector $\bar{v}\in\R_{>0}^n$ such
that $\mathcal{V} := \left\{v\in\mathbb{R}^n:  -\bar{v}\leq v\leq\bar{v}\right\}$, where vector inequalities are interpreted componentwise.

Explicit time arguments are suppressed whenever their omission does not
lead to ambiguity. They are retained whenever necessary to distinguish
among delayed, current, estimated, and predicted quantities. Unless
otherwise stated, a quantity written without an explicit time argument is
understood to be evaluated at the time determined by its definition and the
surrounding context.

\subsection{Ideal Newton-Based Quadratic Program}

The {\bf ideal} Newton-based quadratic program (QP) with exact knowledge of $\ell$ and $h$, no probing, and no output delay is formulated as follows:
\begin{align}
 \label{eq:qp1}
 v^\ast=\operatorname*{arg\,min}_{v,\delta\geq0}
 &\frac12\left\lVert v-v_\esc^{\mathrm{exact}}(t)\right\rVert^2
 +\frac12m\delta^2\\
 \tag{\ref{eq:qp1}a}\label{eq:qp1a}
\hspace{-5mm}\text{s.t.}\qquad \left(G(t)\right)^\top v&\le -\gamma(\|G(t)\|^2)+\delta,\\
 \tag{\ref{eq:qp1}b}\label{eq:qp1b}
 \left(G_h(t)\right)^\top v&\ge -\alpha(z(t)),\\
 \tag{\ref{eq:qp1}c}\label{eq:qp1c}
 -\bar{v}&\le v\le \bar{v},
\end{align}
where 
\begin{align}
\dot\theta&=v,\\
v_\esc^{\mathrm{exact}}(t)&=-K\left(H(t)\right)^{-1}{G}(t),
\end{align}
and $m>0$ is the penalty assigned to the nonnegative relaxation variable
$\delta$. The exact values are represented by $G(t):=\nabla\ell(\theta(t))$, $H(t):=\nabla^2\ell(\theta(t))$, $G_h(t):=\nabla h(\theta(t))$, and $z(t):=h(\theta(t))$.

The function $\gamma:\mathbb{R}_+\rightarrow\mathbb{R}_+$ is a
class-$\mathcal{K}$ function, while $\alpha$ is an extended
class-$\mathcal{K}$ function defined on an interval containing the relevant
barrier values. Constraint \eqref{eq:qp1a} constitutes a CLF-type condition that promotes convergence
toward the optimizer, whereas the CBF constraint \eqref{eq:qp1b} enforces
safety. The relaxation variable is introduced only in the convergence
constraint so that safety retains priority when the two requirements
conflict.

The overall analysis is developed in three stages:
{\bf i)} the delay-free safe Newton-based ES is analyzed to establish the nominal
practical stability and forward invariance properties;
{\bf ii)} predictors are constructed for the objective map and its gradient and Hessian,
barrier value, and barrier gradient. The prediction errors are
characterized in terms of the delay, probing parameters, and estimation
errors; and
{\bf iii)} the delayed safe Newton-based ES is analyzed using strengthened CLF and
CBF constraints that explicitly compensate for prediction errors. Predicted
quantities are denoted by the superscript $\pred$.

\section{Delay-Free Safe Newton-Based ES}\label{secao3}

The ideal Newton-based QP \eqref{eq:qp1} represents a
benchmark formulation in which the probing signal is absent, and the exact objective and barrier derivatives required by the optimization problem are assumed to be available. 
In practice, however, the maps $\ell$ and $h$ are unknown and only their
measured outputs are available. Consequently, the gradients and Hessian
matrices required by the ideal Newton-based QP cannot be evaluated directly.
The purpose of this section is therefore to replace the unavailable exact
derivative information with estimates generated through perturbation-based extremum seeking.
This requires the introduction of a probing signal $S(t)$ and leads to the
implementable \emph{safe Newton-based extremum seeking} (SANES) algorithm. Accordingly,
\begin{align}
    \theta(t)		&=\hat\theta(t)+S(t),\\
    \dot\theta(t) 	&= u+\dot S(t),
\end{align}
where $\dot{\hat\theta}(t)=u$ denotes the slowly varying ES center update.
Thus, convergence is naturally analyzed with respect to the slowly varying
state $\hat\theta(t)$, whereas safety must be enforced with respect to the
actual parameter trajectory $\theta(t)$.

The delay-free case considered here provides the nominal stability and
safety properties required for the subsequent predictor-based analysis.
Lyapunov arguments are used to establish practical convergence, while CBF
arguments are used to establish robust safety. Averaging techniques and time-scale
separation are used to quantify the accuracy of the ES-generated
gradient, Hessian, and inverse-Hessian estimates.

\subsection{Newton-Based Extremum Seeking}

\begin{figure}
\centering
\includegraphics[width=\columnwidth, clip]{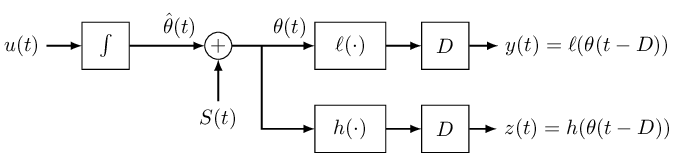}
\caption{Augmented system configuration with an auxiliary control input $u(t)$ and a perturbation input $S(t)$.}
\label{fig:syswithu}
\end{figure}
A block diagram of the augmented system, including the probing signal and
the integrator dynamics associated with the parameter update, is shown in
Fig.~\ref{fig:syswithu}. The probing signal is
$S(t)=[S_i(t)]_{i=1}^n\in\R^{n\times 1}$, where
\begin{equation}
S_i(t)=a_i\sin(\omega_i t),
\end{equation}
and $a_i>0$ and $\omega_i>0$ denote the amplitude and frequency,
respectively, of the $i$th component of the probing signal for
$i\in\{1,2,\ldots,n\}$. The amplitudes are selected sufficiently small so
that the perturbation-induced oscillations and the associated approximation
errors remain within prescribed bounds. For convenience, define $a=\left[a_1~a_2~\ldots~a_n\right]^\top$ and
$\omega=\left[\omega_1~\omega_2~\ldots~\omega_n\right]^\top$.

For the delay-free case, i.e., $D=0$, the nominal ES update is generated
using the Newton-based ES algorithm
\begin{align}
\label{eq:newtonESu}
\dot{\hat\theta}(t)&= u_\esc(t), \quad u_\esc(t)=-K\Gamma(t) \hat{G}(t),\\
\dot{\hat{G}}(t)&=-\omega_l\hat{G}(t)+\omega_l M(t)y^\hp(t)\\
\dot{\Gamma}(t)&=\omega_r\Gamma(t)-\omega_r\Gamma(t)\hat{H}(t)\Gamma(t)\\
\dot{\hat{H}}(t)&=-\omega_l\hat{H}(t)+\omega_lN(t)y^\hp(t)\\
\label{eq:newtonESdY}
\dot{y}^\hp(t)&=-\omega_h y^\hp(t)+\dot{y}(t),
\end{align}
where $\omega_l,\omega_h,\omega_r>0$ are filter corner frequencies, and $K\in\mathbb{R}^{n\times n}$ is a positive-definite gain matrix. The quantities $\hat G(t)$ and $\hat H(t)$ denote estimates of the
objective gradient and Hessian, respectively, while $\Gamma(t)$ denotes an
estimate of the inverse Hessian. The signal $y^\hp(t)$ denotes the
high-pass-filtered objective measurement.

\begin{assumption}
\label{assumption3}
The ES center $\hat{\theta}(t)$ evolves on a slower time scale than the
probing signals. The probing frequencies satisfy the nonresonance
conditions required by the demodulation scheme and are sufficiently
separated from the filter and adaptation bandwidths. See
condition~(42) of~\cite{ghaffari2025second}. In particular, the
ratios of $\omega_h$, $\omega_l$, and $\omega_r$ to the relevant probing
frequencies and frequency separations are sufficiently small that the
resulting averaging and filter-tracking errors are uniformly bounded by a
prescribed constant. Moreover, $a_i\omega_i<\bar{v}_i$ for all $i\in\{1,2,\ldots,n\}$ to avoid control saturation.
\end{assumption}

The demodulation signals are collected in
$M(t)=[M_i(t)]_{i=1}^n\in\R^{n\times 1}$ and
$N(t)=[N_{ij}(t)]_{i,j=1}^{n}\in\R^{n\times n}$, where
\begin{align}
M_i(t)&=\frac{2}{a_i}\sin(\omega_i t)\\
N_{ij}(t)&=\left\{\begin{array}{lcc}
-\dfrac{4}{a_i a_j}\cos\left((\omega_i+\omega_j)t\right),&~&i\neq j\\
-\dfrac{8}{a_i^2}\cos\left(2\omega_it\right),&~&i=j.
\end{array}\right.,
\end{align}
for $i,j\in\{1,2,\ldots,n\}$. 

A block diagram of the delay-free Newton-based ES
algorithm is shown in Fig.~\ref{fig:nes}.
\begin{figure}
\centering
\includegraphics[width=\columnwidth,clip]{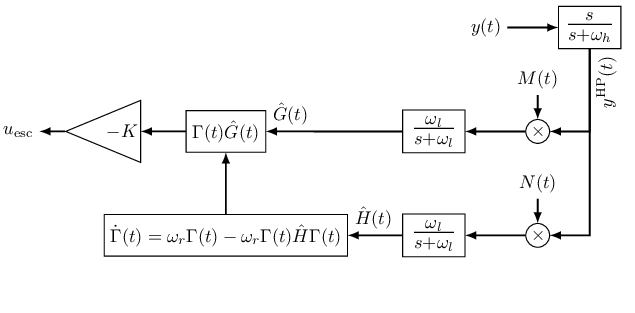}
\vspace{-10mm}
\caption{Conventional Newton-based extremum-seeking loop without constraints and without delay.}
\label{fig:nes}
\end{figure}

\subsection{Estimator Accuracy and Practical Stability of Newton-Based ES}

Recall $\ell(\theta)=\ell^\ast+L(\theta)$, and introduce translated coordinates such that $\theta^\ast=0$.
Since subtraction of the constant $\ell^\ast$ does not affect the
gradient, Hessian, or ES dynamics, one may also set $\ell^\ast=0$ without
loss of generality. By Assumption~\ref{assumption1}, $L$ is positive
definite and strongly convex on $\Theta$, with
\begin{align}
 \nabla L(\theta)\Big\lvert_{\theta=0}&=0,\\
 \nabla^2L(\theta)&\succ0.
\end{align}
More precisely,
${m_2} I\preceq\nabla^2L(\theta)\preceq {M_2} I$ for all
$\theta\in\Theta$.

The slowly varying ES center is denoted by $\hat\theta(t)$, whereas the
parameter applied to the unknown map is
$\theta(t)=\hat\theta(t)+S(t)$. The integrator and high-pass filter
dynamics are
\begin{align}
 \dot{\hat\theta}(t)&=-K\Gamma(t)\hat{G}(t)\\
 \dot{y}^\hp(t)&=-\omega_h y^\hp(t)+\dot{y}(t), \qquad y^\hp(0)=0
 \label{eq:dYdt}
\end{align}
where $\dot{y}(t)=\left(G(t)\right)^\top\dot\theta(t)$ with
$G(t):=\nabla\ell(\theta(t))$. 
The high-pass filter attenuates the DC and slowly varying components of
the measured output while retaining the oscillatory components induced by
the probing signals. These oscillatory components are subsequently
demodulated to obtain estimates of the objective gradient and Hessian.

The low-pass filters satisfy
\begin{align}
 \dot{\hat G}(t)&=-\omega_l\hat G(t)+\omega_lM(t)y^\hp(t), \quad \hat{G}(0)=0\\
 \dot{\hat H}(t)&=-\omega_l\hat H(t)+\omega_lN(t)y^\hp(t), \quad\hat{H}(0)\succ0\\
 \dot\Gamma(t)&=\omega_r\Gamma(t)-\omega_r\Gamma(t)\hat H(t)\Gamma(t),~ \Gamma(0)=\left(\hat{H}(0)\right)^{-1}.
\end{align}

Recall the linear time-invariant system
$\dot{x}(t)=Ax(t)+Bu(t)$, whose solution is
\begin{equation}
x(t)=\e^{At}x(0)+\int_0^t\e^{A(t-\tau)}Bu(\tau)\dif\tau.
\label{eq:LTIsol}
\end{equation}
Application of \eqref{eq:LTIsol} to the gradient and Hessian filters gives
\begin{align}
 \label{eq:hatG}
 \hat G(t)&=\int_0^t \e^{-\omega_lI(t-\tau)}
 \omega_lM(\tau)y^\hp(\tau)\dif\tau,\\
  \label{eq:hatH}
 \hat H(t)&=\e^{-\omega_lIt}\hat H(0)
 +\int_0^t \e^{-\omega_lI(t-\tau)}
 \omega_lN(\tau)y^\hp(\tau)\dif\tau,\\
  \label{eq:Gamma}
 \Gamma(t)&=\left[\e^{-\omega_rIt}\hat H(0)
 +\int_0^t \e^{-\omega_rI(t-\tau)}
 \omega_r\hat H(\tau)\dif\tau\right]^{-1},
\end{align}
where $I$ is the $n\times n$ identity matrix. Equation~\eqref{eq:Gamma}
follows by introducing $Q(t):=\Gamma^{-1}(t)$, which satisfies
$\dot Q=-\omega_rQ+\omega_r\hat H$. Define the transient time
\begin{equation}
 t_{\mathrm{tr}}=\frac{n_\mathrm{tr}}{\min\{\omega_h,\omega_l,\omega_r\}},
\end{equation}
where $n_\mathrm{tr}\in\R_+$ is large enough to provide the desired attenuation
of the exponentially decaying filter transients.  
The initial Hessian estimate is selected positive definite and, whenever
sufficient prior curvature information is available, may be chosen such
that $\hat H(0)\succ m_H H(\theta(0))\succ0$ for any $m_H\gg1$, where
\begin{equation}
 H(\theta(0))=\left.\nabla^2\ell(\theta) \right|_{\theta=\theta(0)}.
\end{equation}
This ordering is a convenient initialization choice and is not required
for the definition of the estimator.

%For periodic probing signals satisfying the required nonresonance
%conditions stated by Assumption~\ref{assumption3}, the averaged demodulated signals satisfy
%\begin{align}
% \frac1\Pi\int_{t}^{t+\Pi} M(\tau)y^\hp(\tau)\dif \tau
% &=\nabla\ell(\hat\theta)+[\bigO(\|a\|^2)]_{n\times1},\\
% \frac1\Pi\int_{t}^{t+\Pi} N(\tau)y^\hp(\tau)\dif \tau
% &=\nabla^2\ell(\hat\theta)+[\bigO(\|a\|^2)]_{n\times n},
%\end{align}
%for $t\ge0$ with additional residual terms arising from finite time-scale separation
%and filter transients. Denote
%$\Pi>0$ as a common probing period satisfying $\omega_i\Pi\in2\pi\mathbb{Z}$ for all $i\in\{1,2,\ldots,n\}$,
%where $\mathbb{Z}$ is the set of integers.

\begin{remark}
The probing amplitudes and frequencies must be selected jointly. 
Reducing the amplitudes decreases the steady-state dither and Taylor
approximation errors. Increasing the separation between the probing
frequencies and the filter bandwidths improves the averaging
approximation. However, note that $\dot{\theta}(t)=\dot{\hat\theta}(t)+\dot{S}(t)$, where $\dot{S}_i(t)=a_i\omega_i\cos(\omega_i t)$ for $i\in\{1,2,\ldots,n\}$. Therefore, excessively large probing frequencies increase $\|\dot S(t)\|$ and may violate the prescribed parameter-rate constraints.
\end{remark}

Using \eqref{eq:LTIsol}, the solution of \eqref{eq:dYdt} is
\begin{equation}
 y^\hp(t)=\e^{-\omega_ht}y^\hp(0)+\int_0^t \e^{-\omega_h(t-\tau)}\dot y(\tau)\dif\tau.
 \label{eq:Ysol}
\end{equation}
Let $p=\e^{-\omega_h(t-\tau)}$ and
$\dif q=\dot y(\tau)\dif\tau$. Integration by parts applied to
\eqref{eq:Ysol} gives
\begin{align}
 y^\hp(t)&=y(t)+e^{-\omega_ht}\bigl(y^\hp(0)-y(0)-\bigr)\nonumber\\
             &\quad -\omega_h\int_0^t \e^{-\omega_h(t-\tau)}y(\tau)\dif\tau.
 \label{eq:Y1}
\end{align}
Since $y(t)=\ell^\ast+L(\theta(t))$,
\begin{align}
\label{eq:intedy}
\int_0^t \e^{-\omega_h(t-\tau)}y(\tau)\dif \tau &=
\frac{\ell^\ast}{\omega_h}\left(1-\e^{-\omega_h t}\right)+\nonumber\\
 & +\int_0^t \e^{-\omega_h(t-\tau)}L(\theta(\tau))\dif\tau.
\end{align}
Substituting \eqref{eq:intedy} into \eqref{eq:Y1} gives
\begin{align}
 y^\hp(t)&=L(\theta(t))+\e^{-\omega_ht}\bigl(y^\hp(0)-L(\theta(0))\bigr)-\nonumber\\
 &\quad-\omega_h\int_0^t \e^{-\omega_h(t-\tau)}L(\theta(\tau))\dif\tau.
 \label{eq:Y2}
\end{align}

Since
$\ell\in C^3(O_\Theta)$, a second-order Taylor expansion of
$L(\hat\theta+S)$ about $\hat\theta$ gives
\begin{align}
L(\theta(t))=L(\hat\theta(t))+\sum_{i=1}^{n}\frac{\partial L(\theta)}{\partial\theta_i}\Big\lvert_{\theta=\hat{\theta}(t)} S_i(t)+\nonumber\\
+\frac{1}{2}\sum_{i=1}^{n}\sum_{j=1}^{n}\frac{\partial^2 L(\theta)}{\partial\theta_i\partial\theta_j}\Big\lvert_{\theta=\hat{\theta}(t)}S_i(t)S_j(t)+R(t),
\label{eq:taylor}
\end{align}
where
$|R(t)|\le(M_3/6)\|S(t)\|^3\le(M_3/6)\|a\|^3$, provided the line segment
joining $\hat\theta(t)$ and $\theta(t)$ remains within the region on which
the third-derivative bound in Assumption~\ref{assumption1} holds.

Using \eqref{eq:taylor} and Assumption~\ref{assumption3}, write
\begin{align}
\int_0^t\e^{-\omega_h(t-\tau)}L(\theta(\tau))\dif\tau
=\frac{1}{\omega_h}L(\hat\theta(t))-\nonumber\\
-\frac{\e^{-\omega_h t}}{\omega_h}L(\hat\theta(0))
+r_h(t),
\label{eq:eL}
\end{align}
where $r_h(t)\le M_{r_h}(\|a\|/\|\omega\|)$ includes decaying oscillations. By
Assumption~\ref{assumption3}, $r_h(t)$ is uniformly bounded after the
transient interval and can be reduced by an appropriate choice of probing
and filter parameters.

Substitution of \eqref{eq:eL} into \eqref{eq:Y2}, together with
$\theta(0)=\hat\theta(0)$ and $y^\hp(0)=0$, yields
\begin{equation}
y^\hp(t)=L(\theta(t))-L(\hat\theta(t))+r_y(t),
\end{equation}
where $r_y(t)=-\omega_h r_h(t)$. Since
$\nabla^k\ell(\theta)=\nabla^kL(\theta)$ for all derivative orders used in
the analysis, the Taylor expansion gives
\begin{align}
 y^\hp(t) &=  \left(\left.\nabla \ell(\theta)\right|_{\theta=\hat\theta}\right)^\top S(t)
 +\frac{1}{2}S(t)^\top \!\! \left.\nabla^2\ell(\theta)\right|_{\theta=\hat\theta}S(t)+\nonumber\\
 &+\bigO(\|a\|^3)+r_y(t).
\end{align}

Noting that Assumption~\ref{assumption3} holds, 
\eqref{eq:hatG}--\eqref{eq:Gamma} give
\begin{align}
 \label{eq:hatGss}
 \hat G(t)&=\left.\nabla \ell(\theta)\right|_{\theta=\hat\theta(t)}
 +[\bigO_{\hat G}(\|a\|^2,\varepsilon_{\rm tr})]_{n\times1},\\
\label{eq:hatHss}
 \hat H(t)&=\left.\nabla^2 \ell(\theta)\right|_{\theta=\hat\theta(t)}
 +[\bigO_{\hat H}(\|a\|^2,\varepsilon_{\rm tr})]_{n\times n},\\
  \label{eq:Gammass}
 \Gamma(t)&=
 \left(\left.\nabla^2 \ell(\theta)\right|_{\theta=\hat\theta(t)}\right)^{-1}
 +[\bigO_{\Gamma}(\|a\|^2,\epsilon_{\rm tr})]_{n\times n},
\end{align}
with representative bounds
\begin{align}
\label{eq:OG}
 \left\|\bigO_{\hat G}(\|a\|^2,\epsilon_{\rm tr})\right\|
 &\le M_{\hat{G}}\|a\|^2+c_{\hat{G}}\epsilon_{\rm tr},\\
 \label{eq:OH}
 \left\|\bigO_{\hat H}(\|a\|^2,\epsilon_{\rm tr})\right\|
 &\le M_{\hat{H}}\|a\|^2+c_{\hat{H}}\epsilon_{\rm tr},\\
 \left\|\bigO_{\Gamma}(\|a\|^2,\epsilon_{\rm tr})\right\|
 &\le M_{\Gamma}\|a\|^2+c_{\Gamma}\epsilon_{\rm tr},
\end{align}
where ${M}_{\star}$ and ${c}_{\star}$ for $\star\in\{\hat{G}, \hat{H}, \Gamma\}$ are positive real constants and $\epsilon_{\rm tr}$ is an exponentially decaying term, corresponding to the impact of the transient response due to the slow filter dynamics. In other words, $\epsilon_{\rm tr}\to 0$ as $t\to\infty$.
Thus, after the transient interval, the ES algorithm provides uniformly
bounded estimates of the objective gradient, Hessian, and inverse Hessian
evaluated at the slow ES center $\hat\theta(t)$.

Next, practical stability of the delay-free Newton-based ES
dynamics is established. Because $S(t)$ is persistent, the actual parameter
$\theta(t)$ cannot, in general, converge exactly to
$\theta^\ast$ for nonzero probing amplitudes. Therefore, 
the slowly varying state $\hat\theta(t)$ is analyzed, and subsequently the impact of the
additional dither amplitude is inspected.

Define
$G_c(t):=\nabla L(\hat\theta(t))=\nabla\ell(\hat\theta(t))$. Let
$K=kI$, where $k>0$. Then $K\Gamma=k\Gamma$ is symmetric positive definite
whenever $\Gamma$ is symmetric positive definite. Omitting the time
argument when no ambiguity arises, the time derivative of the map gives
\begin{align}
\label{eq:Ldot}
 \dot L(\hat\theta)&=\nabla L(\hat\theta)^\top\dot{\hat\theta}\\
  &=G_c^\top\bigl(-K\Gamma\hat G\bigr)\\
 &=G_c^\top\Bigl[-K\Gamma
 \bigl(G_c+[\bigO_{\hat G}(\|a\|^2,\epsilon_{\rm tr})]_{n\times1}\bigr)\Bigr]\\
 &=-G_c^\top K\Gamma G_c
 -G_c^\top K\Gamma
 [\bigO_{\hat G}(\|a\|^2,\epsilon_{\rm tr})]_{n\times1}.
\end{align}

Denote $\rho_{\hat{G}}:=M_{\hat G}\|a\|^2+c_{\hat G}\epsilon_{\rm tr}$.
Then
\begin{align}
\left|
G_c^\top K\Gamma
[\bigO_{\hat G}(\|a\|^2,\epsilon_{\rm tr})]_{n\times1}
\right| \le \|K\Gamma\|\rho_{\hat{G}}\|G_c\|\\
\dot L(\hat\theta) \le -\lambda_{\min}(K\Gamma)\|G_c\|^2
+\|K\Gamma\|\rho_{\hat{G}}\|G_c\|.
\end{align}
Consequently,
\begin{align}
\label{eq:Ldotineq}
 \dot L(\hat\theta)
 &\le-\|G_c\|
 \left(
 \lambda_{\min}(K\Gamma)\|G_c\|
 -\|K\Gamma\|\rho_{\hat{G}}
 \right)\\
 &\le0
 \quad\text{whenever}\quad
 \|G_c\|\ge
 \frac{\|K\Gamma\|}{\lambda_{\min}(K\Gamma)}\rho_{\hat{G}} .
\end{align}
Hence, the Lyapunov function is nonincreasing outside an explicitly bounded
neighborhood of the stationary point. Let
$\underline\lambda_{K\Gamma}>0$ and
$\overline\lambda_{K\Gamma}<\infty$ denote uniform lower and upper bounds,
respectively, on the eigenvalues of $K\Gamma(t)$. Then
\begin{equation}
\lambda_{\min}(K\Gamma)\|G_c\|
-\|K\Gamma\|\rho_{\hat{G}}\ge0
\end{equation}
outside the corresponding ultimate neighborhood, and
\begin{equation}
\label{eq:Gbound}
 \limsup_{t\to\infty}\|G_c(t)\|
 \le
 \frac{\overline\lambda_{K\Gamma}}
      {\underline\lambda_{K\Gamma}}\rho_{\hat{G}}.
\end{equation}

Since $K=kI$ and $\Gamma$ is positive definite,
\begin{align}
\label{eq:lambdamin}
\lambda_{\min}(K\Gamma)&=
k\lambda_{\min}(\Gamma)\\
\label{eq:lambdaminmax}
\lambda_{\min}(\Gamma)&=
\frac{1}{\lambda_{\max}(\Gamma^{-1})}.
\end{align}
After the inverse-Hessian transient, $\Gamma^{-1}$ remains in a bounded
neighborhood of $\nabla^2\ell(\hat\theta)$. Since
${m_2} I\preceq\nabla^2\ell(\hat\theta)\preceq {M_2} I$,
sufficiently small Hessian-estimation and inverse-tracking errors imply
uniform positive lower and upper bounds on $\Gamma$ and hence on
$K\Gamma$.

Using \eqref{eq:gradient_to_parameter_bound}, \eqref{eq:Gbound} yields
\begin{equation}
\label{eq:thetabound}
 \limsup_{t\to\infty}\|\hat\theta(t)\|
 \le
 \frac{1}{{m_2}}
 \frac{\overline\lambda_{K\Gamma}}
      {\underline\lambda_{K\Gamma}}\rho_{\hat{G}} .
\end{equation}
Since
$\theta(t)=\hat\theta(t)+S(t)$ and
$\|S(t)\|\le\|a\|$, the actual parameter satisfies the corresponding
ultimate bound obtained by adding $\|a\|$ to the right-hand side of
\eqref{eq:thetabound}.

Accordingly, the practical-stability and estimation properties can be
summarized as
\begin{align}
 \limsup_{t\to\infty}\|\theta(t)\|
 &\le
 \frac{\overline\lambda_{K\Gamma}}
      {{m_2}\underline\lambda_{K\Gamma}}\rho_{\hat{G}}+\|a\|\\
 \limsup_{t\to\infty}L(\theta(t))
 &\le  \frac{{M_2}}{2}  \left(  \frac{\overline\lambda_{K\Gamma}}
      {m_2\underline\lambda_{K\Gamma}}\rho_{\hat{G}}+\|a\|
 \right)^2\\
 \limsup_{t\to\infty}\|\hat G(t)\|
 &\le \bar{M}_{\hat{G}}\|a\|^2+\bar{c}_{\hat G}\epsilon_{\rm tr}\\
 \limsup_{t\to\infty}\|\hat H(t)-H^\ast\|
 &\le \bar{M}_{\hat{H}}\|a\|^2+\bar{c}_{\hat H}\epsilon_{\rm tr}\\
 \limsup_{t\to\infty}\|\Gamma(t)-(H^\ast)^{-1}\|
 &\le
 \bar{M}_{\Gamma}\|a\|^2+\bar{c}_{\Gamma}\epsilon_{\rm tr},
\end{align}
where $\bar{M}_{\star}$ and $\bar{c}_{\star}$ for $\star\in\{\hat{G}, \hat{H}, \Gamma\}$ are positive real constants.

The preceding results are summarized as follows.

\begin{theorem}
Consider the dynamic system
\begin{align}
\dot{\theta}(t)	&= v(t),\\
y(t) 			&=\ell(\theta(t)),
\end{align}
where $\ell:O_\Theta\rightarrow\mathbb{R}$ satisfies
Assumption~\ref{assumption1}, with
$\ell(\Theta)\subseteq\mathcal{Y}$,
$\Theta\subset\mathbb{R}^n$ compact and convex,
$\mathcal{Y}\subset\mathbb{R}$, and $v(t)\in\mathcal{V}$, with $\mathcal{V} \subset\mathbb{R}^n$ compact and convex. Suppose that
$v(t)=u_\esc(t)+\dot{S}(t)$, where $u_\esc(t)$ is obtained from the Newton-based extremum-seeking algorithm \eqref{eq:newtonESu}--\eqref{eq:newtonESdY} with
$K=kI$, $k>0$. In the translated coordinates
$\ell^\ast=\theta^\ast=0$, and for probing amplitudes and frequencies
satisfying Assumption~\ref{assumption3}, the following properties hold after a finite transient time
$t_{\mathrm{tr}}>0$: {\bf(i)} the gradient, Hessian, and inverse-Hessian
estimates remain in uniformly bounded neighborhoods of
$\nabla\ell(\hat\theta(t))$,
$\nabla^2\ell(\hat\theta(t))$, and
$\left(\nabla^2\ell(\hat\theta(t))\right)^{-1}$, respectively, with bounds that vanish
as $\|a\|\to0$ and
$\epsilon_{\rm tr}\to0$; and {\bf(ii)} the slow state
$\hat\theta(t)$ and the actual perturbed parameter $\theta(t)$ are
uniformly ultimately bounded and satisfy
\begin{equation}
\limsup_{t\to\infty}
\left|\ell(\theta(t))-\ell^\ast\right|
    \le c_y\epsilon^\ast,
\end{equation}
and
\begin{equation}
\limsup_{t\to\infty}
\left\|\theta(t)-\theta^\ast\right\|
    \le c_\theta\epsilon^\ast,
\end{equation}
where $c_y,c_\theta>0$, 
and $\epsilon^\ast>0$ depends on the probing amplitudes and the
time-scale-separation error. Moreover,
$\epsilon^\ast\to0$ as
$\|a\|\to0$ and
$\epsilon_{\rm tr}\to0$. Consequently, the ultimate
neighborhoods of $\ell^\ast$ and $\theta^\ast$ can be made arbitrarily
small by jointly reducing the probing amplitudes and the averaging and
filter-tracking errors, subject to the stated admissibility and
nonresonance conditions.
\end{theorem}

\begin{remark}
The preceding result establishes practical rather than strict asymptotic
stability for nonzero probing amplitudes. In particular,
\eqref{eq:Ldotineq} guarantees strict decrease of $L(\hat\theta)$ outside
an explicitly bounded neighborhood of the optimizer. More generally, a CLF-like 
 condition of the form
\begin{equation}
G_c(\hat\theta)^\top\dot{\hat{\theta}}(t)
\le-\gamma\left(\|G_c(\hat\theta)\|^2\right),
\end{equation}
where $\gamma(\cdot)$ is a class-$\mathcal{K}$ function, provides a
sufficient decrease condition for the slow ES dynamics. When derivative
estimation errors or a relaxation variable are present, suitable
robustness or relaxation terms must be incorporated, leading in general
to practical stability.
\end{remark}

\subsection{Delay-Free SANES Formulation}

The preceding analysis was carried out in translated coordinates for
clarity. The results extend directly to $\theta^\ast\neq0$ through the
coordinate transformation $\tilde\theta=\theta-\theta^\ast$. Likewise,
subtracting the constant $\ell^\ast$ from the objective does not affect its
gradient, Hessian, or the ES update law.
\begin{figure}
\centering
\includegraphics[width=\columnwidth, clip]{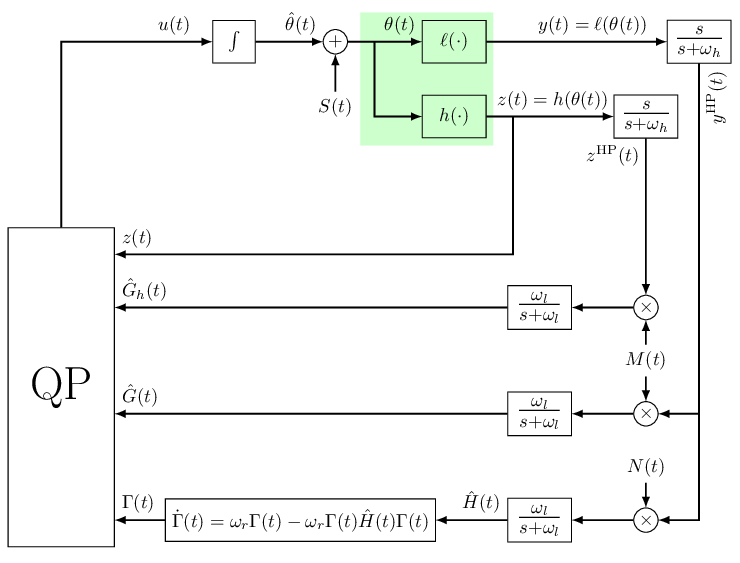}
\caption{Delay-free SANES algorithm.}
\label{fig:safenenod}
\end{figure}

To implement the safety constraint, an estimate
of the barrier gradient is also required because the barrier map $h$ is
unknown. Assuming that only measurements of 
$h(\theta)$ are available, the barrier gradient can be estimated using a
high-pass, demodulation, and low-pass filtering structure analogous to that
used for the objective map. Recall that the quantities $G(t)$, $H(t)$, and $G_h(t)$, appearing in the ideal
Newton-based QP \eqref{eq:qp1}, are exact derivatives and must therefore be
distinguished from the ES estimates
$\hat G(t)$, $\hat H(t)$, and $\hat G_h(t)$. The formulation
\eqref{eq:qp1} serves only as an ideal reference. The delay-free safe Newton-based extremum seeking (SANES)
developed below replaces the unavailable exact derivative information with
bounded-error estimates generated from measurements of $\ell$ and $h$.

The objective gradient estimate is naturally associated with the slow ES
center. Safety, however, concerns the actual
parameter $\theta(t)$. The results are established for the slow ES center, and then the safety condition is modified to ensure robust safety. Denote $G_c(t):=\nabla\ell(\hat\theta(t))$ and $G_{h,c}(t):=\nabla h(\hat\theta(t))$. One can show that
the corresponding estimation errors satisfy
\begin{align}
\label{eq:eGc}
G_c(t) &=\hat{G}(t)+\tilde{G}(t),
\quad \|\tilde{G}(t)\|\le\epsilon_G,\\
\label{eq:eGhc}
{G}_{h,c}(t) &= \hat{G}_h(t)+\tilde{G}_h(t),
\quad \|\tilde{G}_h(t)\|\le\epsilon_{G_h},
\end{align}
where $\hat G(t)$ and $\hat G_h(t)$ denote the corresponding ES gradient
estimates and
\begin{align}
\epsilon_G=M_G\|a\|^2+c_G\epsilon_{\rm tr}\\
\epsilon_{G_h}=M_{G_h}\|a\|^2 +c_{G_h}\epsilon_{\rm tr},
\end{align}
where $\epsilon_{\rm tr}$ is exponentially decayling. 
The term with $\|a\|$ accounts for the difference
between the slow ES center and the actual perturbed parameter and may be
replaced by any sharper uniform bound available for the specific
estimator.

The delay-free SANES architecture is shown in
Fig.~\ref{fig:safenenod}. Unlike the ideal Newton-based QP
\eqref{eq:qp1}, which assumes $S(t)=0$ and exact derivative
information, the delay-free SANES formulation explicitly accounts for both the
probing signal and derivative-estimation errors as stated below
\begin{align}
\label{eq:QPnoD}
u^\ast=\argmin_{u,\delta\ge0}
 &\frac12\left\lVert u-u_{\mathrm{esc}}(t)\right\rVert^2+\frac12m\delta^2\\
\tag{\ref{eq:QPnoD}a}\label{eq:QPnoDa}
\text{s.t.}\qquad
 \left(\hat{G}(t)\right)^\top u &\le -\gamma\left(\|\hat{G}(t)\|^2\right) -\epsilon_{\rm clf}+\mu_{\rm clf}+\delta,\\
\tag{\ref{eq:QPnoD}b}\label{eq:QPnoDb}
 \left(\hat{G}_h(t)\right)^\top u &\ge -\alpha\left(z(t)\right)+\epsilon_{\rm cbf}+\mu_{\rm cbf},\\
\tag{\ref{eq:QPnoD}c}\label{eq:QPnoDc}
 -\bar{v} &\le u +\dot{S}(t) \le \bar{v},
\end{align}
where $m>0$, $\epsilon_{\rm clf}=\epsilon_G\bar{U}$, $\epsilon_{\rm cbf}=\epsilon_{G_h}\bar{U}+2C_\alpha M_{h,1}\|a\|$, $\mu_{\rm clf},\mu_{\rm cbf}\in\R_+$, and
$\bar{U}$ denotes a known uniform bound on the slow update
$\|u(t)\|$. The quantity
$z(t)=h(\theta(t))$ is directly measured in the delay-free case. The
relaxation variable $\delta$ is introduced only in the convergence
constraint so that the safety constraint remains hard.

If the derivative estimates are exact,
$\epsilon_G=\epsilon_{G_h}=0$, and the probing signal is removed, then \eqref{eq:QPnoD} reduces to the ideal Newton-based QP
\eqref{eq:qp1}. Thus, the delay-free SANES QP is the implementable counterpart of the
ideal Newton-based QP when the exact expressions of $\ell$, $h$, and their derivatives are unavailable.

Note that the delay-free estimation bounds are shown using $\epsilon$ and the prediction error bounds for the delayed case are shown using $\varepsilon$. If $D=0$, the delayed bounds will simplify to the delay-free bounds.

\begin{assumption}
\label{asum_fsb}
The robust CBF constraint and the parameter-update constraints are jointly
feasible. Specifically, for every $t\geq 0$ and every admissible
$\theta(t)$ for which the CBF constraint is imposed, there exists at least
one $u(t)\in\mathbb{R}^n$ satisfying
\begin{equation}
\label{eq:feswithD}
    \left(G_h^{\pred}(t)\right)^\top u(t)     \geq     -\alpha\!\left(h^{\pred}(\hat\theta(t))\right)
    +\varepsilon_{\rm cbf},
\end{equation}
and
\begin{equation}
    -\bar{v}     \leq     u(t)+\dot S(t)     \leq     \bar{v}.
\end{equation}
The superscript $\pred$ denotes quantities
predicted over the delay horizon $D$. In the delay-free case, $D=0$, the
predicted quantities reduce to their corresponding delay-free quantities,
and \eqref{eq:feswithD} reduces to the robust CBF constraint
\eqref{eq:QPnoDb}.
\end{assumption}

The CLF condition ensures that $\hat{\theta}$ enters a ball centered at $\theta^\ast$. Therefore, it is desirable to show that $L(\hat\theta)$ is decreasing outside $\Theta_\epsilon$. Using \eqref{eq:eGc}, the time derivative of $L(\hat\theta(t))$ gives
\begin{align}
\dot L(\hat\theta(t))	&= \left(\nabla L(\hat\theta(t))\right)^\top\dot{\hat\theta}(t) \nonumber\\
					&= \left(G_c(t)\right)^\top u(t) \nonumber\\
					&= \left(\hat{G}(t)+\tilde{G}(t)\right)^\top u(t) \nonumber\\
					&\le \left(\hat{G}(t)\right)^\top u(t) + \left\|\tilde{G}(t)\right\|\|u(t)\|.
\end{align}
Using $\|u(t)\|\le\bar{U}$, \eqref{eq:eGc}, and \eqref{eq:QPnoDa}, one arrives at
\begin{align}
&\dot L(\hat\theta(t)) 	\le -\gamma\left(\|\hat{G}(t)\|^2\right)+\mu_{\rm clf}+\delta, \nonumber\\
					&\le -\gamma\left(\bigl(\max\{0,\|G_c(t)\|-\epsilon_G\}\bigr)^2\right)+\mu_{\rm clf}+\delta,
\end{align}
where the second inequality follows from
$\|\hat G\|\ge\max\{0,\|G_c\|-\epsilon_G\}$ and the monotonicity of
$\gamma$.
After a finite transient, $\dot L(\hat\theta)\le 0$ whenever
\begin{equation}
\|G_c\|\ge\epsilon_G+\sqrt{\gamma^{-1}(\mu_{\rm clf}+\bar\delta)},
\end{equation}
where 
\begin{equation}
\limsup_{t\to\infty} \delta(t)\le \bar\delta.
\end{equation}
Using
\eqref{eq:gradient_to_parameter_bound}, the slow ES center therefore satisfies
\begin{equation}
 \limsup_{t\to\infty}\|\hat\theta(t)-\theta^\ast\| \le \frac{\epsilon_G+\sqrt{\gamma^{-1}(\mu_{\rm clf}+\bar\delta)}}{m_2}.
\end{equation}
Since
$\theta(t)=\hat\theta(t)+S(t)$ and $\|S(t)\|\le\|a\|$,
\begin{equation} 
\limsup_{t\to\infty}\|\theta(t)-\theta^\ast\| \le
\frac{\epsilon_G+\sqrt{\gamma^{-1}(\mu_{\rm clf}+\bar\delta)}}{m_2}+\|a\|.
\end{equation}

If $\bar\delta=0$, the ultimate convergence bound reduces to
\begin{equation} 
\limsup_{t\to\infty}\|\theta(t)-\theta^\ast\| \le
\frac{\epsilon_G+\sqrt{\gamma^{-1}(\mu_{\rm clf})}}{m_2}+\|a\|.
\end{equation}

Safety follows from an analogous robust argument. From
\eqref{eq:eGhc}, the time derivative of the CBF gives
\begin{align}
\dot{h}(\hat\theta(t))	&= \Bigl(\nabla h(\hat\theta(t))\Bigr)^\top\dot{\hat\theta}(t) \nonumber\\
 						&= \left(\hat{G}_h(t)+\tilde{G}_h(t)\right)^\top u(t) \nonumber\\
						&\ge \left(\hat{G}_h(t)\right)^\top u(t)-\left\|\tilde{G}_h(t)\right\|\|u(t)\|
\end{align}
Using $\|u(t)\|\le\bar{U}$, \eqref{eq:eGhc}, and \eqref{eq:QPnoDb}, one arrives at
\begin{equation}
\label{eq:safety}
\dot{h}(\hat\theta(t)) \ge -\alpha\left(z(t)\right)+2C_\alpha M_{h,1} \|a\|+\mu_{\rm cbf}.
\end{equation}
Using \eqref{eq:hLipschitz} and applying \eqref{eq:alphaLpsch}, one can transform \eqref{eq:safety} into
\begin{equation}
\label{eq:alphaineq1}
\dot{h}(\hat\theta(t)) \ge -\alpha\left(h(\hat\theta(t))\right)+C_\alpha M_{h,1} \|a\|+\mu_{\rm cbf}.
\end{equation}
One can use \eqref{eq:alphaineq1} to show that
\begin{equation}
\alpha\left(h(\hat\theta(t))\right)\ge C_\alpha M_{h,1} \|a\|+\mu_{\rm cbf}
\end{equation} 
for all $\hat\theta\in\Theta_{0,{\rm r}}$, where 
\begin{equation}
\Theta_{0, {\rm r}}:=\left\{\theta\in\Theta_0: h(\theta)\ge\alpha^{-1}(C_\alpha M_{h,1}\|a\|+\mu_{\rm cbf})\right\}.
\end{equation}
Using \eqref{eq:alphaLpsch}, one arrives at
\begin{equation}
\label{eq:alpha_h}
\left|\alpha(h(\theta))-\alpha(h(\hat\theta))\right| \le C_\alpha \left| h(\theta)-h(\hat\theta)\right|
\end{equation}
Inequality \eqref{eq:alpha_h} gives
\begin{align}
\alpha(h(\theta))-\alpha(h(\hat\theta))  &\ge -C_\alpha \left| h(\theta)-h(\hat\theta)\right|\\
					    \alpha(h(\theta)) &\ge \alpha(h(\hat\theta))-C_\alpha M_{h,1} \|a\| \\
\label{eq:alphaineq2}   \alpha(h(\theta)) &\ge \mu_{\rm cbf}
\end{align}
ensuring that $h(\theta(t))\ge0$. Therefore, provided the CBF constraint remains feasible, the robust subset $\Theta_{0,\rm r}$ is forward invariant. Hence, robust safety is guaranteed in the presence of bounded barrier-gradient
estimation error.

The above results are summarized in the following theorem.

\begin{theorem}
\label{thm:SANES_no_delay}
Consider the closed-loop system $\dot{\theta}(t)=u^\ast+\dot{S}(t)$ and $y(t)=\ell(\theta(t))$ with safety measurements $z(t)=h(\theta(t))$, where $u^\ast$ is obtained from QP
\eqref{eq:QPnoD} and the nominal extremum-seeking input
$u_{\mathrm{esc}}$ is given by \eqref{eq:newtonESu}. Suppose that the
objective map $\ell:O_\Theta\rightarrow\mathbb{R}$ satisfies
Assumption~\ref{assumption1}, the control barrier function
$h:O_\Theta\rightarrow\mathbb{R}$ satisfies
Assumption~\ref{assumption2}, the time-scale-separation conditions of
Assumption~\ref{assumption3} hold, and the feasibility condition in
Assumption~\ref{asum_fsb} is satisfied. Assume further that $\theta^\ast\in\operatorname{int}(\Theta_0)$, 
$\theta(0)\in\Theta_{0,\rm r}$, 
and that, after the extremum-seeking transient, the objective- and
barrier-gradient estimation errors satisfy $\|\tilde G(t)\|\leq\epsilon_G$ and 
$\|\tilde G_h(t)\|\leq\epsilon_{G_h}$.
Suppose that $\epsilon_G$ and $\epsilon_{G_h}$ are sufficiently small for
the resulting closed-loop trajectory to remain in the admissible region
where the preceding assumptions and error bounds hold.

Then, the actual
parameter trajectory satisfies
\begin{equation}
    \limsup_{t\to\infty} \|\theta(t)-\theta^\ast\| \leq \frac{\epsilon_G+\sqrt{\gamma^{-1}(\bar\delta+\mu_{\rm clf})}}{m_2}+\|a\|,
\end{equation}
where $\bar\delta$ is a uniform asymptotic bound of $\delta(t)$ as $t\to\infty$.
Moreover, the safe set $\Theta_{0,{\rm r}}$ is forward invariant under the
resulting closed-loop dynamics.
\end{theorem}

% ============================================================
% Part I: Introduction and Predictor Development
% ============================================================

\section{Delayed Safe Newton-Based ES}  \label{secao4}

When the output measurements are subject to a constant delay $D>0$, the
current quantities required by the Newton-based QP \eqref{eq:qp1} are no longer directly
available. In particular, the objective gradient and Hessian, the barrier
value, and the barrier gradient available from measurements correspond to
the delayed parameter $\theta(t-D)$ rather than the current parameter
$\theta(t)$. Consequently, predictor equations are required to construct
the current quantities from delayed measurements and derivative estimates. 
To distinguish the total prediction-error bounds from their delay-free estimation counterparts, 
the former are denoted by $\varepsilon$, whereas the latter are denoted by $\epsilon$.

\subsection{Predictor Design}

Recall that the slowly varying ES state, the actual perturbed parameter,
and the delayed objective measurement satisfy
\begin{align}
\label{eq:hattheta}
 \dot{\hat\theta}(t)&=u(t),\\
 \theta(t)&=\hat\theta(t)+S(t),\\
 y(t)&=\ell(\theta(t-D)).
\end{align}
Since $y(t+D)=\ell(\theta(t))$, the objective is to predict the quantities associated
with $\theta(t)$ using information available from the delayed state
$\theta(t-D)$. Integrating \eqref{eq:hattheta} over the interval
$[t-D,t]$ gives
\begin{align}
\hat\theta(t)&=\hat\theta(t-D)+\Delta(t),\\
\Delta(t)&=\int_{t-D}^t u(\tau)\dif\tau.
\end{align}
Hence,
\begin{equation}
    \theta(t) = \hat\theta(t-D)+\Delta(t)+S(t).
\end{equation}
The increment $\Delta(t)$ describes the displacement of the slow ES center
over one delay interval, whereas $S(t)$ accounts for the instantaneous
probing perturbation.

\begin{remark}
Since $\ell,h\in C^3(O_\Theta)$ and $\Theta$ is compact, the
derivatives required for the Taylor expansions are uniformly bounded on
the admissible set. Furthermore, the prediction displacement satisfies
\begin{equation}
    \|\Delta(t)\|
    =
    \left\|\int_{t-D}^{t}u(\tau)\dif\tau\right\|
    \le \epsilon_\Delta, \quad \epsilon_\Delta=D\bar{U},
\end{equation}
where $\bar{U}$ is the prescribed bound on the parameter estimate update rate.
Hence, for sufficiently small delay $D$ and bounded inputs, the first-order
Taylor expansions used in the predictor construction have bounded
remainders. The resulting prediction errors are explicitly incorporated
into the robustness margins of the subsequent CLF and CBF constraints.
\end{remark}

Recall $G(t):=\nabla\ell(\theta(t))$ and $H(t):=\nabla^2\ell(\theta(t))$.
A first-order Taylor expansion of the objective gradient and Hessian about
the delayed slow state yields
\begin{align}
\label{eq:Gpred1}
{G}(t) &= G_c(t-D)+H_c(t-D)\Delta(t)+\nonumber\\
&+[\bigO(\|a\|,\epsilon_\Delta^2)]_{n\times 1},\\
{H}(t) &= H_c(t-D)+\left(\mathcal{T}_c(t-D)\right)^\top\Delta_D(t)+\nonumber\\
&+[\bigO(\|a\|,\epsilon_\Delta^2)]_{n\times n},
\end{align}
where
\begin{align}
G_c(t-D)=\nabla \ell(\theta)\Big\lvert_{\theta=\hat{\theta}(t-D)}\\
H_c(t-D)=\nabla^2 \ell(\theta)\Big\lvert_{\theta=\hat\theta(t-D)}
\end{align}
and $\mathcal{T}_c(t-D)$ represents the third-order derivative of
$\ell$ arranged in a matrix form compatible with the Hessian prediction.
Specifically,
\begin{align}
\mathcal{T}(t)& =\left.\begin{bmatrix}
\nabla H_{11} 	& \nabla H_{12} & \cdots & \nabla H_{1n}\\
\nabla H_{21} 	& \nabla H_{22} & \cdots & \nabla H_{2n}\\
\vdots 			& \vdots 		& \ddots & \vdots\\
\nabla H_{n1} 	& \nabla H_{n2} & \cdots & \nabla H_{nn}
\end{bmatrix}\right\lvert_{\theta=\hat\theta(t-D)}\\
 \Delta_D(t) &=\begin{bmatrix}
 \Delta(t) & 0_{n\times 1} 			& \cdots & 0_{n\times 1}\\
 0_{n\times 1} 		& \Delta(t) 		& \cdots & 0_{n\times 1}\\
 \vdots & \vdots 		& \ddots & \vdots\\
 0_{n\times 1} 		& 0_{n\times 1} 			& \cdots & \Delta(t)
 \end{bmatrix}.
\end{align}
Here, $H_{ij}$ denotes the $(i,j)$th entry of the Hessian, and
$\mathcal{T}(t),\Delta_D(t)\in\mathbb{R}^{n^2\times n}$. The
$\bigO(\|a\|,\epsilon_\Delta^2)$ terms account for the difference between the slow ES center
and the actual perturbed parameter.

Similarly, let
\begin{equation}
    G_h(t):=\nabla h(\theta(t)),
    \qquad
    H_h(t):=\nabla^2 h(\theta(t)).
\end{equation}
Taylor expansions of the barrier value and its gradient yield
\begin{align}
\label{eq:htrue}
 h(\theta(t))&=h(\theta(t-D))+\left(G_{h,c}(t-D)\right)^\top\Delta(t)+\nonumber\\
 &+\bigO(\|a\|),\\
 G_h(t)&=G_{h,c}(t-D)+H_{h,c}(t-D)\Delta(t)+\nonumber\\
 &+[\bigO(\|a\|)]_{n\times 1},
\end{align}
where
\begin{align}
G_{h,c}(t-D)=\nabla h(\theta)\Big\lvert_{\theta=\hat{\theta}(t-D)}\\
H_{h,c}(t-D)=\nabla^2 h(\theta)\Big\lvert_{\theta=\hat\theta(t-D)}.
\end{align}
The Hessian matrices in the preceding expressions are symmetric because
$\ell$ and $h$ are $C^3$.

The quantities in \eqref{eq:Gpred1} and \eqref{eq:htrue} are unavailable
because the maps and their derivatives are unknown. Replacing these
quantities by their ES estimates gives the implementable predictor
\begin{align}
\label{eq:Gpred2}
G^\pred(t)&=\hat{G}(t-D)+\hat{H}(t-D)\Delta(t),\\
H^\pred(t)&=\hat{H}(t-D)+\left(\hat{{\mathcal{T}}}(t-D)\right)^\top\Delta_D(t),\\
\dot{\Gamma}^\pred(t)&=\omega_r\Gamma^\pred(t)
-\omega_r\Gamma^\pred(t) H^\pred(t)\Gamma^\pred(t),\\
\label{eq:hpred}
h^\pred(t)&=h(\theta(t-D))+\left(\hat{G}_h(t-D)\right)^\top\Delta(t),\\
\label{eq:Ghpred2}
G_h^\pred(t)&= \hat{G}_h(t-D)+\hat{H}_h(t-D)\Delta(t).
\end{align}
The notation $\pred$ distinguishes predicted quantities from exact and
estimated quantities. The predicted inverse Hessian is generated through
the Riccati-type adaptation law rather than by directly inverting
$H^\pred(t)$.

The predicted Newton-based ES input is
\begin{equation}
\label{eq:u_pred}
 u_{\mathrm{esc}}^\pred(t)
 =-K\Gamma^\pred(t)G^\pred(t),\quad K\succ0.
\end{equation}

To estimate the third-order derivative required in the Hessian predictor,
define the demodulation matrix
$P(t)=[P_{ij}(t)]_{n^2\times n}$, where
\begin{equation}
P_{ij}(t)=[P_{ij1}(t)~~P_{ij2}(t)~\cdots~P_{ijn}(t)]^\top
\end{equation} 
and
\begin{align}
P_{ijk}(t)=-\frac{8b_P}{a_ia_ja_k}
\sin\Big(\left(\omega_i+\omega_j+\omega_k\right)t\Big).
\end{align}
This definition holds for all
$i,j,k\in\{1,2,\ldots,n\}$, where $b_P=3!$ when $i=j=k$,
$b_P=2!$ when exactly two indices coincide, and $b_P=0!$ when the three
indices are distinct. The probing frequencies are selected to satisfy
Assumption~\ref{assumption3}, ensuring that the first-, second-,
and third-order derivatives are recovered with bounded averaging errors~\cite{ghaffari2025second}.
\begin{figure}
\centering
\includegraphics[width=\columnwidth,clip]{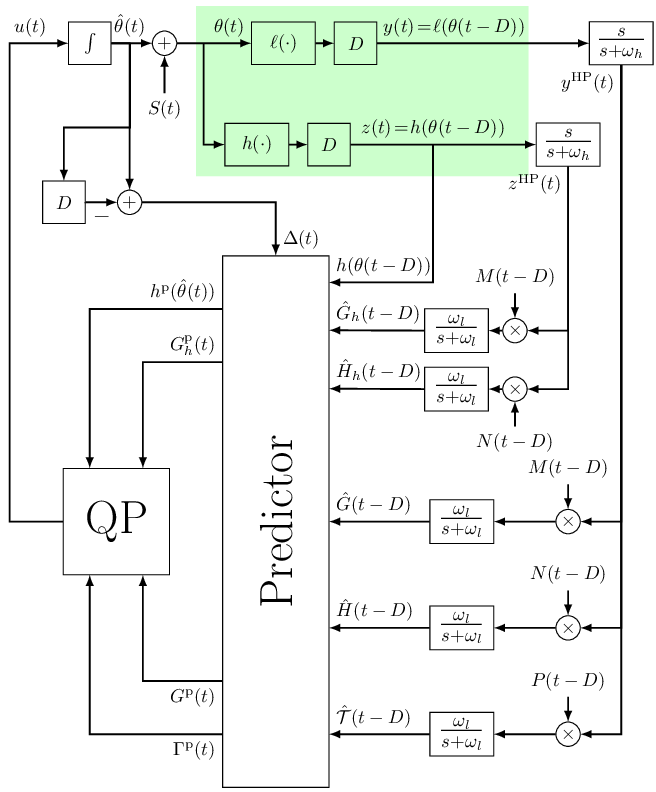}
\caption{Delayed SANES algorithm.}
\label{fig:safenewithd}
\end{figure}

% ============================================================
% Part II: Delayed SANES Formulation
% ============================================================

\subsection{Delayed SANES Formulation}

Figure~\ref{fig:safenewithd} illustrates the delayed SANES architecture.
Incorporating the predicted quantities into the ideal Newton-based QP
produces an implementable formulation based on delayed measurements and
bounded-error derivative estimates. The optimization variable $u$ represents the slow ES update
$\dot{\hat\theta}(t)$, where $\|u(t)\|\le\bar{U}$,
where $\bar{U}$ is a known finite constant. 

The delayed SANES QP is formulated as
\begin{align}
\label{eq:QPwithD}
u^\ast &=\argmin_{u,\delta\ge0}  \frac12\left\lVert u-u_{\mathrm{esc}}^\pred(t)\right\rVert^2+\frac12m\delta^2,\\
 \tag{\ref{eq:QPwithD}a}\label{eq:QPwithDa}
\text{s.t.}~ \left(G^\pred(t)\right)^\top u &\le -\gamma\left(\| G^\pred(t)\|^2\right)-\varepsilon_{\rm clf}+\mu_{\rm clf}+\delta,\\
 \tag{\ref{eq:QPwithD}b}\label{eq:QPwithDb}
\left(G^\pred_h(t)\right)^\top  u &\ge -\alpha\left( h^\pred(t)\right)+ \varepsilon_{\rm cbf}+\mu_{\rm cbf},\\
 \tag{\ref{eq:QPwithD}c}\label{eq:QPwithDc}
u_{\mathrm{esc}}^\pred(t)  &=-K\Gamma^\pred(t) G^\pred(t)\\
 \tag{\ref{eq:QPwithD}d}\label{eq:QPwithDd}
-\bar{v} &\le u+\dot S(t) \le \bar{v},
\end{align}
where $m>0$, $\varepsilon_{\rm clf}=\varepsilon_G\bar{U}$, $\varepsilon_{\rm cbf}=C_\alpha\varepsilon_h + \varepsilon_{G_h}\bar{U}+2C_\alpha M_{h,1}\|a\|$, and $\mu_{\rm clf}, \mu_{\rm cbf}\in\R_+$.

The objective in \eqref{eq:QPwithD} keeps the QP solution close to the
predicted nominal Newton-based ES input. Constraint \eqref{eq:QPwithDa} is a
robustified CLF-type condition for practical convergence. The relaxation
variable $\delta$ appears only in this constraint so that stability may be
softened if it conflicts with safety or input feasibility. Constraint
\eqref{eq:QPwithDb} is a strengthened predicted CBF condition. The terms
$C_\alpha\varepsilon_h$ and
$\varepsilon_{G_h}\bar{U}$ compensate for bounded errors in the
predicted barrier value and barrier gradient, respectively.
If the predicted quantities coincide with the corresponding true current
quantities and the probing signal is removed, \eqref{eq:QPwithD} reduces
to the ideal Newton-based QP \eqref{eq:qp1}.

% ============================================================
% Part III: Practical Stability Analysis
% ============================================================

\subsection{Practical Stability Analysis}

Recall that $y(t)=\ell^\ast+L(\theta(t-D))$,
where $L(\theta)=\ell(\theta)-\ell^\ast$ is positive definite with respect
to $\theta^\ast$. As in the previous case, the analysis focuses on the slow ES center. The current shifted objective
$L(\hat\theta(t))$ is therefore selected as the Lyapunov candidate. 

Denote $G_c(t):=\nabla\ell(\hat\theta(t))$ and $G_{h,c}(t):=\nabla h(\hat\theta(t))$, where $\hat\theta(t)$ is the exact ES center. 
Define the objective-gradient prediction error by
\begin{equation}
\label{eq:eG}
    e_G(t) := G_c(t) - G^\pred(t).
\end{equation}
Combining \eqref{eq:Gpred1}, \eqref{eq:Gpred2},
\eqref{eq:hatGss}, and \eqref{eq:hatHss}, together with the Taylor
remainder bounds, imply that there exists a known finite constant
$\varepsilon_G>0$ satisfying $\|e_G(t)\|\le \varepsilon_{G}$. 
The constant $\varepsilon_G$ depends on the delay, the derivative-estimation
errors, the Taylor remainder, and the probing amplitudes. In particular,
one may write a representative bound of the form
\begin{equation}
    \varepsilon_G \le \epsilon_G  +c_{G,D1}\epsilon_\Delta
    +c_{G,D2}\epsilon_\Delta^2    +c_{G,a}\|a\|,
\end{equation}
where $\epsilon_G$ is the delay-free ES gradient-estimation bound and the
remaining constants are finite on the compact admissible set. Hence,
$\varepsilon_G\to\epsilon_G$ as $D\to0$.

Since $G_c(t)=G^\pred(t)+e_G(t)$, 
the derivative of the Lyapunov candidate satisfies
\begin{align}
    \dot L(\hat\theta(t))
    &= \left(\nabla L(\hat\theta(t))\right)^\top \dot{\hat{\theta}}(t) \nonumber\\
    &= \left(G^\pred(t)+e_G(t)\right)^\top u(t)\nonumber\\
    &\le \left(G^\pred(t)\right)^\top u(t)+\|e_G(t)\|\|u(t)\|\nonumber\\
    \label{eq:Ldot1}
    &\le \left(G^\pred(t)\right)^\top u(t) +\varepsilon_G\bar{U}
\end{align}
Using CLF constraint \eqref{eq:QPwithDa}, \eqref{eq:Ldot1} is transformed to
\begin{align}
    \dot L(\hat\theta(t)) &\le -\gamma\!\left(\|G^\pred(t)\|^2\right)+\delta \nonumber\\
    \label{eq:Ldotineq2}
        &\le -\gamma\!\left(\bigl(\max\{0,\|G_c(t)\|-\varepsilon_G\}\bigr)^2\right)+\mu_{\rm clf}+\delta,
\end{align}
where \eqref{eq:Ldotineq2} follows from
\begin{equation}
    \|G^\pred(t)\| = \|G_c(t)-e_G(t)\| \ge \max\{0,\|G_c(t)\|-\varepsilon_G\}
\end{equation}
and the monotonicity of the class-$\mathcal K$ function $\gamma$.
The derivative is negative whenever
\begin{equation}
\gamma\!\left(\bigl(\max\{0,\|G_c(t)\|-\varepsilon_G\}\bigr)^2\right) > \delta+\mu_{\rm clf}.
\end{equation}
Consequently,
\begin{equation}
\limsup_{t\to\infty}\|G_c(t)\| \le \varepsilon_G+\sqrt{\gamma^{-1}(\mu_{\rm clf}+\bar\delta)},
\end{equation}
where 
\begin{equation}\limsup_{t\to\infty} \delta(t)\le \bar\delta.
\end{equation}

Using \eqref{eq:gradient_to_parameter_bound}, every trajectory ultimately
saftifies
\begin{equation}
\limsup_{t\to\infty}\|\theta(t)-\theta^\ast\|\le \frac{\varepsilon_G+\sqrt{\gamma^{-1}(\mu_{\rm clf}+\bar\delta)}}{m_2} +\|a\|.
\end{equation}
In the special case $\bar\delta=0$, this reduces to
\begin{equation}
    \limsup_{t\to\infty}\|\theta(t)-\theta^\ast\| \le \frac{\varepsilon_G+\sqrt{\gamma^{-1}(\mu_{\rm clf})}}{m_2}+\|a\|.
\end{equation}
This establishes practical stability, or uniform ultimate boundedness, of
the delayed closed-loop system. Moreover, because
$\varepsilon_G\to\epsilon_G$ as $D\to$, the ultimate
bound continuously approaches the bound obtained in the delay-free case.

% ============================================================
% Part IV: Safety Proof
% ============================================================

\subsection{Robust Safety Under Bounded Prediction Errors}

Suppose that the predicted barrier value and gradient satisfy
\begin{align}
    e_h(t)&:= h(\theta(t))- h^\pred(t),\\
    \label{eq:GhPred}
    e_{G_h}(t)&:= G_{h,c}(t)-G_h^\pred(t),
\end{align}
where the prediction errors are bounded according to $|e_h(t)| \le \varepsilon_h$, $\|e_{G_h}(t)\| \le \varepsilon_{G_h}$.
Using \eqref{eq:GhPred}, the derivative of the true barrier function at the slow ES center
satisfies
\begin{align}
    \dot h(\hat\theta(t)) 	&= \left(\nabla h(\hat\theta(t))\right)^\top \dot{\hat\theta}(t)\\
    						&= \left(G_h^\pred(t)+e_{G_h}(t)\right)^\top u(t)\\
    						&\ge \left(G_h^\pred(t)\right)^\top u(t) -  \|e_{G_h}(t)\| \|u(t)\|\\
    						&\ge \left(G_h^\pred(t)\right)^\top u(t) - \varepsilon_{G_h}\bar{U}.
\end{align}
Applying \eqref{eq:QPwithDb} gives
\begin{equation}
    \dot h(\hat\theta(t)) \ge -\alpha\!\left(h^\pred (t)\right) +C_\alpha\varepsilon_h+2C_\alpha M_{h,1}\|a\|+\mu_{\rm cbf}.
\end{equation}
Moreover, \eqref{eq:htrue} and \eqref{eq:hpred}, together with the
estimation and Taylor-remainder bounds, imply that
\begin{equation}
    |h^\pred(t)-h(\theta(t))| \le \varepsilon_h.
\end{equation}
By the Lipschitz continuity of $\alpha$,
\begin{equation}
    \alpha\!\left(h^\pred(t)\right) \le \alpha\!\left(h(\theta(t))\right) + C_\alpha\varepsilon_h.
\end{equation}
Therefore,
\begin{equation}
\label{eq:hdotineq1}
    \dot h(\hat\theta(t)) \ge -\alpha\!\left(h(\theta(t))\right)+2C_\alpha M_{h,1}\|a\|+\mu_{\rm cbf}.
\end{equation}
Using \eqref{eq:hLipschitz} and applying \eqref{eq:alphaLpsch}, one can transform \eqref{eq:hdotineq1} into
\begin{equation}
\dot h(\hat\theta(t)) \ge -\alpha\!\left(h(\hat\theta(t))\right)+C_\alpha M_{h,1} \|a\|+\mu_{\rm cbf},
\end{equation}
which is the same as \eqref{eq:alphaineq1}. Therefore, the safe set $\Theta_{0,{\rm r}}$ is forward invariant.

The above results are summarized as follows.
\begin{theorem}
Consider the closed-loop system $\dot{\theta}(t)=u^\ast+\dot{S}(t)$ and $y(t)=\ell(\theta(t-D))$ with safety measurements $z(t)=h(\theta(t-D))$, where $D>0$ is a constant delay, $u^\ast$ is obtained from QP
\eqref{eq:QPwithD}, and the predictor given by \eqref{eq:Gpred2}--\eqref{eq:u_pred}. Suppose that the
objective map $\ell:O_\Theta\rightarrow\mathbb{R}$ satisfies
Assumption~\ref{assumption1}, the control barrier function
$h:O_\Theta\rightarrow\mathbb{R}$ satisfies
Assumption~\ref{assumption2}, the time-scale-separation conditions of
Assumption~\ref{assumption3} hold, and the feasibility condition in
Assumption~\ref{asum_fsb} is satisfied. Assume further that $\theta^\ast\in\operatorname{int}(\Theta_0)$, 
$\theta(0)\in\Theta_{0,\rm r}$, and that, after the extremum-seeking transient, the objective-gradient, barrier, and 
barrier-gradient prediction errors satisfy $\|e_G(t)\|\leq\varepsilon_G$, $|e_h(t)|\le\varepsilon_h$, and 
$\| e_{G_h}(t)\|\leq\varepsilon_{G_h}$, respectively. Suppose that $\varepsilon_G$, $\varepsilon_h$, and $\varepsilon_{G_h}$ are sufficiently small for
the resulting closed-loop trajectory to remain in the admissible region
where the preceding assumptions and error bounds hold.

Then, the actual parameter trajectory satisfies
\begin{equation}
    \limsup_{t\to\infty}     \|\theta(t)-\theta^\ast\| \leq \frac{\varepsilon_G+\sqrt{\gamma^{-1}(\mu_{\rm clf}+\bar\delta)}}{m_2}+\|a\|,
\end{equation}
where $\bar\delta$ is a uniform asymptotic bound of $\delta(t)$ as $t\to\infty$.
Moreover, the safe set $\Theta_{0,{\rm r}}$ is forward invariant under the
resulting closed-loop dynamics.
\end{theorem}

For completeness, the following feasibility result is presented, which is applicable to both the delay-free and delayed SANES QP since the delay-free SANES QP \eqref{eq:QPnoD} is a special case of~\eqref{eq:QPwithD}, where $D=0$.
\begin{proposition}
\label{prop:QP_wellposed}
Suppose Assumption~\ref{asum_fsb} holds and $m>0$. Then, at
every $\theta\in\Theta_{0, \rm r}$ and every time $t$ under consideration, the
SANES QP~\eqref{eq:QPwithD} is feasible and admits a unique optimal solution
$(u^\ast(t),\delta^\ast(t))$.
\end{proposition}

\begin{proof}
By Assumption~\ref{asum_fsb}, there exists an admissible
control input satisfying the hard CBF and parameter-update constraints.
For any such input, the relaxed CLF constraint can always be satisfied by
selecting a sufficiently large $\delta\ge0$. Hence, the feasible set of
the quadratic program is nonempty.

All constraints are affine in the decision variables $(u,\delta)$ and
therefore define a closed and convex feasible set. Furthermore, for
$m>0$, the Hessian of the objective function with respect to
$(u,\delta)$ is
\begin{equation}
    \begin{bmatrix}
        I & 0\\
        0 & m
    \end{bmatrix}
    \succ0.
\end{equation}
Hence, the objective function is strictly convex and coercive. Therefore,
the quadratic program admits a unique minimizer.
\end{proof}

\section{A Case Study: Obstacle Avoidance}\label{secao5}

Source seeking in robotics is a potential application of the proposed algorithm, in which the objective function and the safety map may represent various physical variables (such as distance, sound intensity, light intensity, electromagnetic field strength, or scent concentration). This example focuses on cases where physical distance is the intended variable. In real-world applications, these functions are unknown, but their measurements are available. The objective and safety maps adopted below can be chosen to satisfy
Assumptions~\ref{assumption1} and \ref{assumption2} on the prescribed compact set $\Theta$.

In this example, the objective is to find a target located at $\theta^\ast=[2~~4]^\top$, where the distance map is represented by
\begin{align} 
\ell(\theta)&=\ell^\ast+\frac12(\theta-\theta^\ast)^\top H^\ast(\theta-\theta^\ast)+\nonumber\\
 &+\frac1{12}\left((\theta-\theta^\ast)^\top(\theta-\theta^\ast)\right)^2,
\end{align}
where $\ell^\ast=100$, and 
\begin{equation}
H^\ast=\begin{bmatrix}10&3\\3&2\end{bmatrix}.
\end{equation}
One can show that $\nabla\ell(\theta^\ast)=0$ and $\nabla^2\ell(\theta)$ and $\nabla^3\ell(\theta)$ are continuous and bounded on any prescribed compact set $\Theta$.

Consider an obstacle located at $\theta_o=[1~~2]^\top$ and the CBF given by
\begin{equation}
 h(\theta)=\frac12(\theta-\theta_o)^\top(\theta-\theta_o)-\frac{1}{2}r_o^2,
\end{equation}
where $r_o=1$ is the safety radius; that is, the parameter trajectory must avoid a circle with radius $r_o$ centered at $\theta_o$. One can show that $h(\theta)$ and $\nabla h(\theta)$ are continuous and Lipschitz on $\Theta$. Also, note that $\alpha(\xi)=b_\alpha\xi$  is an extended class-$\mathcal{K}$ function for $\xi\in[-c_1, c_2]$ and $b_\alpha,c_1,c_2>0$ and $\gamma(\xi)=b_\gamma\xi$ is a class-$\mathcal{K}$ function for $\xi\in[0,c_2]$ and $b_\gamma>0$. Consider $D=5$. Time is measured in seconds.

The probing vector is $S(t)=0.1[\sin(70t)~~\sin(50t)]^\top$, $\omega_l = 1$, $\omega_r = 0.1$, $\omega_h = 1$ rad/s, $K=\mathrm{diag}([0.05~~0.05])$, $\bar{v}=[10~~10]^\top$, $b_\alpha  = 2$, $b_\gamma = 10^{-4}$, $\varepsilon_{\rm clf}-\mu_{\rm clf}=0.1$, and $\varepsilon_{\rm cbf}+\mu_{\rm cbf}=0.15$. Initial conditions are set as $\hat{H}(0)=1000I$, $\Gamma(0)=\left(\hat{H}(0)\right)^{-1}$, and $\theta(0)=[-2~~-2]^\top$.

\begin{figure}
\centering
\includegraphics[width=3in]{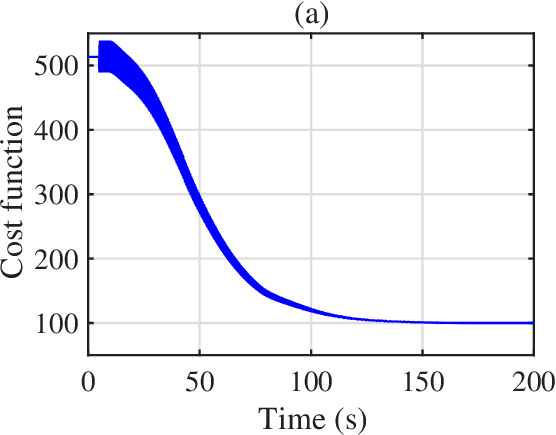}
\caption{Temporal evolution of the objective function. The objective value approaches a neighborhood of its minimum.}
\label{fig:cost}
\end{figure}
\begin{figure}
\centering
\includegraphics[width=3in]{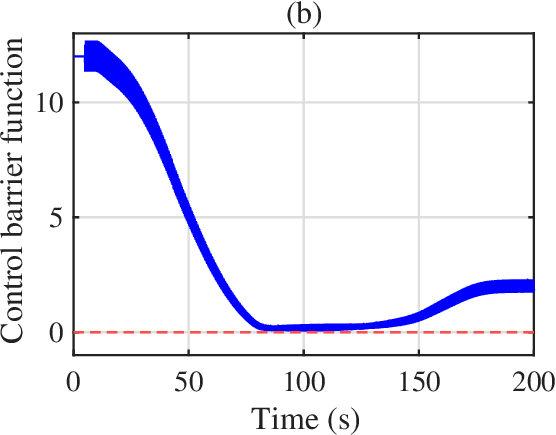}
\caption{Temporal evolution of the control barrier function. The CBF remains positive throughout the simulation, demonstrating satisfaction of the safety constraint.}
\label{fig:cbf}
\end{figure}
\begin{figure}
\centering
\includegraphics[width=3in]{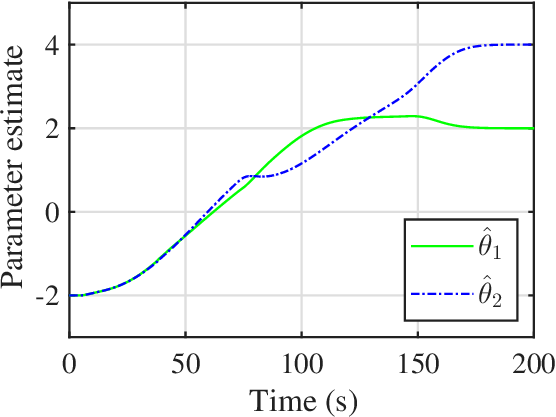}
\caption{The parameter estimates approach a neighborhood of the optimizer.}
\label{fig:param}
\end{figure}
\begin{figure}
\centering
\includegraphics[width=3in]{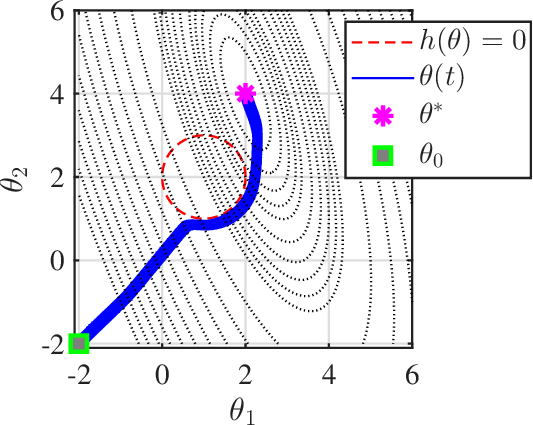}
\caption{The parameter trajectory shows that the safety constraint is satisfied while practical stability is achieved.}
\label{fig:trajectory}
\end{figure}
As shown in Fig.~\ref{fig:cost}--\ref{fig:trajectory}, the algorithm safely approaches the optimal point. The evolution of the CBF, indicating the safety margin relative to the obstacle boundary, is shown in Fig.~\ref{fig:cbf}, where the CBF remains positive. Fig.~\ref{fig:param} shows the parameter estimates (the slow ES center), which approach the optimizer with high accuracy. The parameter trajectory is shown in Fig.~\ref{fig:trajectory}, where the thickness of the trajectory reflects the effect of the dither vector $S(t)$. The algorithm initially follows a nearly straight path toward the optimal point, consistent with the curvature-normalized search direction of Newton-based ES, as shown by the uniform temporal evolution of parameters up to $t=70$~s, after which the trajectory deflects to avoid the obstacle. The SANES QP produces a nearly direct path without active safety constraints, potentially reducing unnecessary deviations in the parameter trajectory.

Increasing $\mu_{\rm clf}$ increases the ultimate bound on the slow ES center, and increasing $\mu_{\rm cbf}$ increases the safety buffer around the obstacle; i.e., the trajectory deviates farther from the obstacle.

\section{Conclusion}\label{secao6}

This work developed a safe Newton-based extremum seeking (SANES) framework for model-free optimization of an unknown static map subject to an unknown safety constraint and delayed measurements. The proposed approach combines Newton-based extremum seeking with robust control Lyapunov function (CLF) and control barrier function (CBF) constraints formulated through quadratic programming. Unlike an ideal Newton-based QP that requires exact knowledge of the objective and safety maps and their derivatives, the SANES formulation relies only on measured objective and safety outputs and employs probing signals and demodulation filters to estimate the derivative information required for optimization and safety enforcement.

The theoretical development considered the delay-free and delayed cases separately. For the delay-free case, practical convergence of the Newton-based extremum-seeking dynamics was established using a Lyapunov analysis over the admissible parameter set. Bounds on the gradient and Hessian estimation errors were explicitly incorporated into the strengthened CLF and CBF constraints. The resulting robustness margins ensure practical convergence despite bounded objective-gradient estimation errors, while preserving the forward invariance of a robust subset of the prescribed safe set under bounded barrier-gradient estimation errors.

For the delayed case, a model-free prediction framework was developed to compensate for the common constant delay affecting the objective and safety measurements. Taylor-based predictors were constructed for the objective gradient and Hessian, barrier value and gradient, and the corresponding Newton-based extremum-seeking input. Bounds on the resulting estimation and prediction errors were used to robustify the delayed CLF--CBF QP. The analysis showed that the closed-loop parameter trajectory ultimately remains within a bounded neighborhood of the unknown optimizer, with the ultimate bound determined by the probing, estimation, and prediction errors. Simultaneously, appropriate robustness margins in the predicted CBF constraint guarantee forward invariance of the robust subset of the safe set. The numerical case study further demonstrates the effectiveness of the proposed framework and its ability to perform safe model-free optimization with delayed measurements.

%\begin{ack}                               % Place acknowledgements
%Partially supported by the Roman Senate.  % here.
%\end{ack}

\bibliographystyle{plain}        % Include this if you use bibtex 
\bibliography{autosam}           % and a bib file to produce the 
                                 % bibliography (preferred). The
                                 % correct style is generated by
                                 % Elsevier at the time of printing.

%\begin{thebibliography}{99}     % Otherwise use the  
                                 % thebibliography environment.
                                 % Insert the full references here.
                                 % See a recent issue of Automatica 
                                 % for the style.
%  \bibitem[Heritage, 1992]{Heritage:92}
%     (1992) {\it The American Heritage. 
%     Dictionary of the American Language.}
%     Houghton Mifflin Company.
%  \bibitem[Able, 1956]{Abl:56}
%     B.~C.~Able (1956). Nucleic acid content of macroscope. 
%     {\it Nature 2}, 7--9. 
%  \bibitem[Able {\em et al.}, 1954]{AbTaRu:54}   
%     B.~C. Able, R.~A. Tagg, and M.~Rush (1954).
%     Enzyme-catalyzed cellular transanimations.
%     In A.~F.~Round, editor, 
%     {\it Advances in Enzymology Vol. 2} (125--247). 
%     New York, Academic Press.
%  \bibitem[R.~Keohane, 1958]{Keo:58}
%     R.~Keohane (1958).
%     {\it Power and Interdependence: 
%     World Politics in Transition.}
%     Boston, Little, Brown \& Co.
%  \bibitem[Powers, 1985]{Pow:85}
%     T.~Powers (1985).
%     Is there a way out?
%     {\it Harpers, June 1985}, 35--47.

%\end{thebibliography}

%\appendix
%\section{A summary of Latin grammar}    % Each appendix must have a short title.
%\section{Some Latin vocabulary}         % Sections and subsections are supported  
                                        % in the appendices.
\end{document}